\documentclass[a4paper,12pt]{article}
\usepackage{test,subcaption,caption,tikz}

\usepackage{hyperref}
\hypersetup{colorlinks,citecolor=blue,urlcolor=magenta}

\usepackage[style=authoryear-comp,
sorting=nyt,
dashed=false, 
maxcitenames=2, 
maxbibnames=99, 
uniquelist=false,
uniquename=false,
giveninits=true, 
natbib, 
date=year 
]{biblatex}

\AtBeginRefsection{\GenRefcontextData{sorting=ynt}}
\AtEveryCite{\localrefcontext[sorting=ynt]}

\DeclareFieldFormat{pages}{#1} 
\renewbibmacro{in:}{\ifentrytype{article}{}{\printtext{\bibstring{in}\intitlepunct}}} 

\usepackage[inline,shortlabels]{enumitem}
\setlist[enumerate,1]{label=(\roman*)}

\usepackage[margin = 1.25in]{geometry}
\usepackage{setspace}
\title{Bubble Economics\thanks{We thank Gadi Barlevy, Narayana Kocherlakota, Jos\'e Scheinkman, and Joseph Stiglitz for discussions, comments, and continued encouragement; Gaetano Bloise for providing detailed comments that significantly improved our understandings; and Jessica Li and Max Yang for excellent research assistance.}} 
\author{Tomohiro Hirano\thanks{Department of Economics, Royal Holloway, University of London and Research Associate at the Center for Macroeconomics at the London School of Economics and the Canon Institute for Global Studies. Email: \href{mailto:tomohiro.hirano@rhul.ac.uk}{tomohih@gmail.com}.} \and Alexis Akira Toda\thanks{Department of Economics, University of California San Diego. Email: \href{mailto:atoda@ucsd.edu}{atoda@ucsd.edu}.}}

\date{\today\\
{\normalsize Forthcoming at \emph{Journal of Mathematical Economics} 50th Anniversary Issue}}

\numberwithin{equation}{section}
\numberwithin{lem}{section}

\begin{document}

\maketitle

\begin{abstract}

This article provides a self-contained overview of the theory of rational asset price bubbles. We cover topics from basic definitions, properties, and classical results to frontier research, with an emphasis on bubbles attached to real assets such as stocks, housing, and land. The main message is that bubbles attached to real assets are fundamentally nonstationary phenomena related to unbalanced growth. We present a bare-bones model and draw three new insights:
\begin{enumerate*}
    \item the emergence of asset price bubbles is a necessity, instead of a possibility;
    \item asset pricing implications are markedly different between balanced growth of stationary nature and unbalanced growth of nonstationary nature; and   
    \item asset price bubbles occur within larger historical trends involving shifts in industrial structure driven by technological innovation, including the transition from the Malthusian economy to the modern economy.
\end{enumerate*}

\medskip


\medskip

\end{abstract}

\section{Introduction}

An asset price bubble is, loosely speaking, a situation in which the asset price is too high to be justified by fundamentals. History abounds with bubbly episodes. \citet[Appendix B]{Kindleberger2000} documents 38 episodes in the 1618--1998 period. \citet*{JordaSchularickTaylor2015} study bubbles in housing and equity markets in 17 countries over the past 140 years. Famous examples are the Dutch tulip mania of the 1630s, the South Sea Bubble of 1720 in England, the Japanese real estate and stock market bubble of the 1980s, and the U.S. dot-com bubble of the late 1990s and the housing bubble of the early 2000s, among others.

It is easy to dismiss these bubbly episodes as ``irrational exuberance'' and exclude them from formal economic analysis, at least those based on rational equilibrium models. However, we could also easily imagine assets that seem expensive relative to fundamentals for a prolonged period, such as real estates in certain cities including London, San Francisco, Sydney, and Vancouver. How should we think of them from an economic theoretical point of view?

The purpose of this article is to introduce the theory of rational asset price bubbles to the general audience. As \citet{Miao2014} acknowledges, ``this topic is typically not taught in macroeconomics or microeconomics'' and ``there are many misunderstandings of some conceptual and theoretical issues''. The situation remains the same after a decade. To deal with these issues, we provide an overview of the theory of rational asset price bubbles in a mostly self-contained way so that non-experts including PhD students can follow the argument. We cover topics from basic definitions, properties, and classical results to frontier research. Since \citet{Miao2014} discusses the theory of rational bubbles up to 2014, which almost exclusively focuses on pure bubbles (assets that pay no dividends such as money), we proceed with more emphasis on subsequent developments, especially bubbles attached to real assets such as stocks, housing, and land.  

The main message is that bubbles attached to real assets are fundamentally nonstationary phenomena related to unbalanced growth. This implies that to understand the essence of asset price bubbles, we need to depart from stationary models with a steady state to nonstationary models without one. We present a bare-bones model with bubbles attached to productive land (a dividend-paying asset) and draw three new insights. 
\begin{enumerate*}
    \item The emergence of asset price bubbles is a necessity, instead of a possibility.
    \item Asset pricing implications are markedly different between balanced growth of stationary nature and unbalanced growth of nonstationary nature.  
    \item Asset price bubbles occur in larger historical trends involving shifts in industrial structure driven by technological innovation, including the transition from the Malthusian economy to the modern economy.
\end{enumerate*}
These insights are wildly different from common results in pure bubble models (bubbles attached to an intrinsically useless asset). At the same time, the bare-bones model we present includes a pure bubble model as a special case. In this sense, it can also provide firmer footing for pure bubble models.

In our opinion, the theory of rational asset price bubbles attached to real assets has the potential to fundamentally change the conventional thinking about asset bubbles. Moreover, it is still largely underdeveloped and hence provides a vast fertile ground for applications. We hope to expand the community of researchers working on this exciting topic.

There are already several surveys on asset price bubbles. \citet{BrunnermeierOehmke2013} discuss a variety of models including heterogeneous beliefs and asymmetric information. In contrast, our focus is rational bubble models in which the asset price $P_t$ exceeds the fundamental value $V_t$ as an equilibrium outcome, even if common knowledge about assets is assumed (fundamental value is unambiguously defined).\footnote{It is often claimed that asset bubbles are associated with human irrationality. Of course, it is an important factor and the rational bubble literature does not deny that. Rather, the literature stresses that even if agents hold rational expectations and common knowledge about assets, bubbles can still arise as an equilibrium outcome. This implies that if we assume human's irrationality, asset price bubbles are more likely to occur.} \citet{Barlevy2018} discusses policy implications and \citet{MartinVentura2018} discuss the quantitative literature and thus are complementary.\footnote{There is also a large literature that studies money as a medium of exchange using a search-theoretic approach, including seminal contributions by \citet{Jevons1875}, \citet{KiyotakiWright1989}, \citet{Iwai1996}, and \citet{LagosWright2005}. For developments in monetary-search models, see \citet{WilliamsonWright2010}. Our review focuses on asset price bubbles in a competitive economy with agents holding rational expectations following the definition in \S\ref{subsec:prelim_defn}, with an emphasis on bubbles attached to real assets. Obviously, the two approaches are complementary.} 

\subsection{Japan's bubble economy in 1980s}

Partly to put bubbles into a historical context, and partly for entertainment, this section discusses the Japanese real estate bubble in the 1980s, which is one of the most spectacular bubbly episodes in human history.\footnote{For more details, see \citet{Noguchi1994}, \citet*{OkinaShirakawaShiratsuka2001}, and \citet[in Japanese]{Ishii2011}.}

Japan experienced a rapid postwar economic growth. The Japanese economy coped with the oil shocks in the 1970s through \emph{sh\=o-ene}, or efficiency improvement in energy consumption. With a rising standard of living, the Japanese society was filled with optimism. One popular catchphrase was ``Japan as number one''  \citep{Vogel1979}. Several factors contributed to the emergence of the bubble economy. First, the 1985 report by the National Land Agency titled ``Capital Remodeling Plan'' predicted Tokyo was destined to become a global financial hub.\footnote{Specifically, the report predicted that the demand for office space in Tokyo would increase from 3,700ha in 1985 to 8,700ha (equivalent of 250 skyscrapers) by year 2000 to house corporate headquarters and international financial services.} 
Second, in response to a brief recession caused by the rapid appreciation of the yen and a contraction in the manufacturing sector following the Plaza Accord on September 22, 1985, the Bank of Japan cut the official discount rate (from 5.0\% in January 1986 to 2.5\% in February 1987) to stimulate the economy. General optimism and a sustained low interest rate environment fueled land speculation.

According to \citet[Table 1]{Noguchi1990}, land price appreciation accelerated around 1986. As of 1987, the price-rent ratio of the Marunouchi business district in Tokyo was 20 times that of the inner city of London. The term ``bubble'' appeared for the first time in an article by Noguchi in the November 26, 1987 issue of \emph{T\=oy\=o Keizai} titled ``Land Price Inflated with Bubbles''. The easy money also made consumers extravagant: people flocked to fancy restaurants, discos, and ski resorts, drank expensive French wines like Roman\'ee-Conti and Ch\^ateau Latour, and bid up 10,000 yen bills along streets to secure taxi rides.\footnote{The 2007 Japanese science fiction comedy film ``Bubble Fiction: Boom or Bust'' vividly illustrates this situation.} 
In large cities, many office buildings were constructed, some of which were dubbed ``The Tower of Bubble'' due to the postmodern architectural style that was popular at the time (Figure \ref{fig:tower}).\footnote{One of them is the Tokyo Metropolitan Government Building constructed in 1990. Another is ``Joule A'' constructed in 1990 by Azabu Building. Its president, Kitaro Watanabe, the 6th wealthiest individual in the world in 1990 according to Forbes, borrowed 700 billion yen and possessed 165 office buildings in Tokyo and 6 luxury hotels in Hawaii.}

Some cautioned against land speculation. According to an interview article in the October 23, 2000 issue of \emph{Nikkei Business}, Taro Kaneko, former Ministry of Finance bureaucrat and then the president of Marusan Security, wrote a letter in 1988 addressed to employees stating ``The total land value in Japan is estimated to be 4 times of U.S. The land area is 1/25, so the unit price is really 100 times. The land price of the Imperial Palace is about the same as California. Even if the Japanese economy is booming, we cannot expect such an abnormal disparity to be sustainable. Moreover, the Japanese population will decline. Therefore, you should refrain from purchasing housing for the time being''.\footnote{The sources of urban legends that compare Japanese land to California or U.S. are unclear. The earliest in print we found is on p.~103 of the book ``\emph{Chiky\=u Jidai no Shin Shiten}'' (A New Perspective in a Global Era) by business consultant Kenichi Ohmae published in December 1988, who writes ``the land price of the Imperial Palace equals that of the entire California''. In English, \citet[p.~244]{Frankel1991} (whose manuscript circulated in October 1989) notes ``A favorite ``factoid,'' which is apparently true, is that the grounds of the Imperial Palace [\ldots] is worth more than all the land in the State of California'', with evidence attributed to \citet{Boone1989,Boone1990}. \citet[Ch.~1]{Boone1990}, which is a revision of \citet{Boone1989} and provides detailed empirical analysis, states ``In the spring of 1990, the value of Japanese land was estimated to equal fifteen trillion dollars. This [\ldots] is over three times the value of all the land in the United States''.}

By the end of the 1980s, average households could no longer afford land, which became a big social issue known as the ``land problem''. Following the official discount rate hike from 2.5\% to 6.0\% in May 1989--August 1990 by the Bank of Japan and the introduction of the Real Estate Loan Total Quantity Restriction by the Ministry of Finance on March 27, 1990, easy money dried out and the bubble collapsed.

\section{Definition and characterization of bubbles}\label{sec:prelim}

We start the discussion from the definition of asset price bubbles and their characterization, largely based on \citet[\S 2]{HiranoTodaNecessity}.

\subsection{Formal definition}\label{subsec:prelim_defn}

We consider an infinite-horizon, deterministic economy with a homogeneous good and time indexed by $t=0,1,\dotsc$. Consider an asset with infinite maturity that pays dividend $D_t\ge 0$ and trades at ex-dividend price $P_t$, both in units of the time-$t$ good. In the background, we assume the presence of rational, perfectly competitive investors. Free disposal of the asset implies $P_t\ge 0$.\footnote{\label{fn:P>0}If $P_t<0$, by purchasing one additional share of the asset at time $t$ and immediately disposing of it, an investor can increase consumption at time $t$ by $-P_t>0$ with no cost, which violates individual optimality.} Let $q_t>0$ be the Arrow-Debreu price, \ie, the date-0 price of the consumption good delivered at time $t$, with the normalization $q_0=1$. The absence of arbitrage implies
\begin{equation}
    q_tP_t = q_{t+1}(P_{t+1}+D_{t+1}). \label{eq:noarbitrage}
\end{equation}
Iterating the no-arbitrage condition \eqref{eq:noarbitrage} forward and using $q_0=1$, we obtain
\begin{equation}
    P_0=\sum_{t=1}^T q_tD_t+q_TP_T. \label{eq:P_iter}
\end{equation}
Because $q_t>0$ and $D_t\ge 0$, the sequence $\set{\sum_{t=1}^T q_tD_t}$ is increasing in $T$. Furthermore, because \eqref{eq:P_iter} holds and $P_T\ge 0$, the sequence is bounded above by $P_0$. Therefore the sequence converges and the infinite sum of the present value of dividends
\begin{equation*}
    V_0\coloneqq \sum_{t=1}^\infty q_tD_t
\end{equation*}
exists, which is called the \emph{fundamental value} of the asset. Letting $T\to\infty$ in \eqref{eq:P_iter}, we obtain
\begin{equation}
    P_0=\sum_{t=1}^\infty q_tD_t+\lim_{T\to\infty}q_TP_T=V_0+\lim_{T\to\infty}q_TP_T. \label{eq:P0}
\end{equation}
When the last term in \eqref{eq:P0} equals zero, or
\begin{equation}
    \lim_{T\to\infty}q_TP_T = 0, \label{eq:TVC}
\end{equation}
we say that the \emph{transversality condition} for asset pricing holds,\footnote{The term ``transversality condition'' has two meanings: one is the transversality condition \eqref{eq:TVC} for asset pricing and the other is that for optimality in infinite-horizon optimization problems. One should clearly distinguish the two but the meaning should be clear from the context.} in which case \eqref{eq:P0} implies that $P_0=V_0$ and the asset price equals its fundamental value. If $\lim_{T\to\infty}q_TP_T>0$, then $P_0>V_0$, and we say that the asset contains a \emph{bubble}. In other words, an asset price bubble is a situation in which the asset price exceeds its fundamental value defined by the present value of dividends.

\subsection{Remarks}

Some remarks are in order regarding the definition of bubbles. 

First, our definition of asset price bubbles simply follows that in the literature; see for instance \citet[Footnote 8]{Tirole1985}, \citet[Ch.~5]{BlanchardFischer1989}, and \citet[\S 13.6]{Miao2020}. Note that free disposal forces $P_t\ge 0$, so by \eqref{eq:P0} negative bubbles ($P_t<V_t$) are impossible. The term ``bubble'' often appears in the academic literature as well as the popular press. However, the definition is often different and most papers are not about bubbles in the sense of the definition in \S\ref{subsec:prelim_defn} because authors rarely verify the violation of the transversality condition \eqref{eq:TVC}.\footnote{Concerning the discount rate, we obviously need to use it in a consistent manner with the individual optimization problem and the no-arbitrage condition.} In this article, we provide an overview of rational asset price bubbles following the definition of \S\ref{subsec:prelim_defn}.

Second, note that in deterministic economies, for all $t$ we have
\begin{equation}
    P_t=\underbrace{\frac{1}{q_t}\sum_{s=1}^\infty q_{t+s}D_{t+s}}_\text{fundamental value $V_t$}+\underbrace{\frac{1}{q_t}\lim_{T\to\infty} q_TP_T}_\text{bubble component}. \label{eq:Pt}
\end{equation}
Therefore either $P_t=V_t$ for all $t$ or $P_t>V_t$ for all $t$, so the economy is permanently in either the bubbly or the fundamental regime. This is different from the popular conception that bubbles are short-run asset appreciations followed by crashes. According to the mathematical definition and consistent expectations (rational expectations) in deterministic economies, bubbles are a long run phenomenon, namely a \emph{permanent overvaluation of assets}. Indeed, as we explain in \S\ref{subsec:found_necessary} and \S\ref{sec:bare-bones}, it is natural to assume that asset price bubbles occur over a long-term historical process that involves a shift in industrial structure.

Third, and related to the second point, in deterministic economies bubbles are never expected to collapse. However, this should not be taken literally. What the theory tells us is that given the expectations of the agents, permanent bubbles may emerge. Of course, the equilibrium will change if agents revise their expectations. For instance, a Roman living in the 2nd century might have expected the Empire to prosper forever. Another Roman living in the 4th century might have been more pessimistic. The same applies to catchphrases like ``Japan as number one'' in the 1980s or the ``New Economy'' in the 1990s. However, if expectations change, a bubble may collapse. From an ex-post perspective, this bubble may appear to be a short-run phenomenon. We need to separate ex-ante and ex-post. In what follows, we abstract from expectations formation and study asset pricing implications given the expectations.

\subsection{Bubble Characterization Lemma}

In general, checking the transversality condition \eqref{eq:TVC} directly could be difficult because it involves $q_T$, which is generally not easy to evaluate. The following lemma provides an equivalent characterization and facilitates checking the presence or absence of bubbles.\footnote{Theorem 5 of \citet{MontrucchioPrivileggi2001} presents a similar condition under uncertainty but is limited to the ``only if'' part.}

\begin{lem}[Bubble Characterization, \citealp{HiranoTodaNecessity}]\label{lem:bubble}
If $P_t>0$ for all $t$, the asset price exhibits a bubble if and only if $\sum_{t=1}^\infty D_t/P_t<\infty$.
\end{lem}

One important implication of the Bubble Characterization Lemma \ref{lem:bubble} is that bubbles are fundamentally nonstationary phenomena, where prices grow faster than dividends. To see why, note that there is an asset price bubble if and only if the infinite sum of dividend yields $D_t/P_t$ is finite. Because $\sum_{t=1}^\infty 1/t=\infty$ but $\sum_{t=1}^\infty 1/t^\alpha<\infty$ for any $\alpha>1$, roughly speaking, there is an asset price bubble if the price-dividend ratio $P_t/D_t$ grows faster than linearly. In particular, a bubble can never happen if the price-dividend ratio is stationary. This point together with the fact that economists are trained to study stationary models may explain why the theory of asset price bubbles attached to dividend-paying assets is underdeveloped.

\section{Theoretical foundations of bubbles}\label{sec:found}

\subsection{Bubbles in overlapping generations models}\label{subsec:found_OLG}

Bubbles are well known to arise in overlapping generations (OLG) models. Here we present a simple analysis based on \citet{Samuelson1958}. Time is denoted by $t=0,1,\dotsc$. At each time $t$, a new agent is born, who lives for two periods and has the Cobb-Douglas utility function
\begin{equation}
    U(y_t,z_{t+1})=(1-\beta)\log y_t+\beta \log z_{t+1}, \label{eq:CD}
\end{equation}
where $(y_t,z_{t+1})$ is consumption when young and old. At $t=0$, there is also an old agent who only cares about their consumption $z_0$. For any $t$, the endowments are $a>0$ when young and $b>0$ when old. There is also an intrinsically useless asset (like fiat money), which is in unit supply, perfectly durable, and initially owned by the old at $t=0$.

The budget constraints of an agent born at time $t$ are therefore
\begin{subequations}\label{eq:budget_OLG}
\begin{align}
    &\text{Young:} & y_t+P_tx_t &= a, \label{eq:budget_OLG_young}\\
    &\text{Old:} & z_{t+1}&=b+P_{t+1}x_t, \label{eq:budget_OLG_old}
\end{align}
\end{subequations}
where $P_t$ is the price of the asset and $x_t$ is asset holdings. Obviously, because the initial old exit the economy and are endowed with one share of the asset, the budget constraint is
\begin{equation*}
    z_0=b+P_0.
\end{equation*}
As usual, a competitive equilibrium is defined by utility maximization and market clearing.

Let us find all (deterministic) equilibria of this economy. Because the asset pays no dividends, its fundamental value is $V_t=0$. By the remark after \eqref{eq:Pt}, we have either $P_t=0$ for all $t$ or $P_t>0$ for all $t$. If $P_t=0$ for all $t$, the budget constraints imply $(y_t,z_{t+1})=(a,b)$ for all $t$, which is clearly an equilibrium, called \emph{fundamental equilibrium}.

Next, suppose $P_t>0$ for all $t$, which is called a \emph{bubbly equilibrium}. Eliminating $x_t$ from the budget constraints \eqref{eq:budget_OLG}, we obtain
\begin{equation*}
    y_t+\frac{P_t}{P_{t+1}}z_{t+1}=a+\frac{P_t}{P_{t+1}}b.
\end{equation*}
Using the familiar Cobb-Douglas formula, the demand of the young is
\begin{equation}
    y_t=(1-\beta)\left(a+\frac{P_t}{P_{t+1}}b\right). \label{eq:yt_OLG}
\end{equation}
Using the budget constraint of the young \eqref{eq:budget_OLG_young} and the market clearing condition $x_t=1$, the asset price satisfies
\begin{align}
    & P_t=P_tx_t=a-y_t=a-(1-\beta)\left(a+\frac{P_t}{P_{t+1}}b\right) \notag \\
    \iff & \frac{1}{P_{t+1}}=\frac{\beta a}{(1-\beta)b}\frac{1}{P_t}-\frac{1}{(1-\beta)b}, \label{eq:diff_OLG}
\end{align}
which is a first-order linear difference equation in $1/P_t$. If $\beta a=(1-\beta)b$, then \eqref{eq:diff_OLG} implies
\begin{equation*}
    \frac{1}{P_t}=\frac{1}{P_0}-\frac{t}{(1-\beta)b}\to -\infty,
\end{equation*}
which contradicts $P_t>0$. Hence assume $\beta a\neq (1-\beta)b$. Then the general solution to \eqref{eq:diff_OLG} is
\begin{equation*}
\frac{1}{P_t}=\frac{1}{\beta a-(1-\beta)b}+\left(\frac{\beta a}{(1-\beta)b}\right)^t\left(\frac{1}{P_0}-\frac{1}{\beta a-(1-\beta)b}\right).
\end{equation*}
It is easy to show that $1/P_t>0$ for all $t$ if and only if $\beta a>(1-\beta)b$ and $P_0\le \beta a-(1-\beta)b$. Therefore we obtain the following proposition.

\begin{prop}\label{prop:samuelson}
In the OLG model, the following statements are true.
\begin{enumerate}
    \item\label{item:samuelson_f} If $\beta a\le (1-\beta)b$, the unique equilibrium is fundamental.
    \item\label{item:samuelson_b} If $\beta a>(1-\beta)b$, there exists a unique fundamental equilibrium as well as a continuum of bubbly equilibria parameterized by $0<P_0\le \beta a-(1-\beta)b$. The bubbly equilibrium is stationary if $P_0=\beta a-(1-\beta)b$; otherwise, the equilibrium is asymptotically bubbleless in the sense that $P_t\to 0$ as $t\to\infty$.
\end{enumerate}
\end{prop}

An intrinsically useless asset that pays no dividends is often called a \emph{pure bubble}. In pure bubble models, there is always a fundamental equilibrium in which the asset has no value. In addition, there often exist a continuum of bubbly equilibria, as in Proposition \ref{prop:samuelson}. Thus the model suffers from equilibrium indeterminacy.

\subsection{Bubbles in infinite-horizon models}\label{subsec:found_infhor}

OLG models are arguably stylized. The first example of a bubbly equilibrium with infinitely-lived agents is due to \citet{Bewley1980}. In this model, there are two agents with utility $\sum_{t=0}^\infty \beta^t u(c_t)$, where $\beta>0$ is the discount factor and $u$ is the period utility function. For simplicity, assume that $u$ exhibits constant relative risk aversion $\gamma>0$, so
\begin{equation*}
    u(c)=\begin{cases*}
        \frac{c^{1-\gamma}}{1-\gamma} & if $\gamma\neq 1$,\\
        \log c & if $\gamma=1$.
    \end{cases*}
\end{equation*}
The endowment streams of the two agents are given by
\begin{align*}
    &\text{Agent 1:} &&(a,bG,aG^2,bG^3,\dotsc),\\
    &\text{Agent 2:} &&(b,aG,bG^2,aG^3,\dotsc),
\end{align*}
where $a>b$ and $G>0$. Thus, the aggregate endowment $(a+b)G^t$ grows at a constant rate $G>0$, but the income ratio between the two agents alternates between $a/b$ and $b/a$ every period. At time $t$, call the agent with endowment $aG^t$ ($bG^t$) ``rich (poor)''. Suppose that there is a pure bubble asset in unit supply, which is initially held by the poor agent. There is a shortsales constraint, and agents can take a long position in the asset but not a short position.

Let us construct a bubbly equilibrium in this model. Conjecture that the asset price is given by $P_t=pG^t$ for some constant $p>0$, and that every period the poor agent sells the entire share of the asset to the rich agent. Then the consumption allocation of the rich and poor agents is
\begin{equation*}
    (c_t^r,c_t^p)=((a-p)G^t,(b+p)G^t).
\end{equation*}
The first-order condition of the rich agent, which must hold with equality because the agent is buying the asset, is
\begin{align}
    u'(c_t^r)P_t=\beta u'(c_{t+1}^p)P_{t+1}&\iff ((a-p)G^t)^{-\gamma}pG^t=\beta ((b+p)G^{t+1})^{-\gamma}pG^{t+1} \notag \\
    &\iff p=\frac{(\beta G^{1-\gamma})^{1/\gamma}a-b}{1+(\beta G^{1-\gamma})^{1/\gamma}}. \label{eq:foc_rich}
\end{align}
For this $p$ to be positive, we need $b<(\beta G^{1-\gamma})^{1/\gamma}a$. Because the bubble asset cannot be shorted, the first-order condition of the poor agent becomes the inequality
\begin{align}
    u'(c_t^p)P_t\ge \beta u'(c_{t+1}^r)P_{t+1} &\iff ((b+p)G^t)^{-\gamma}pG^t\ge \beta ((a-p)G^{t+1})^{-\gamma}pG^{t+1} \notag \\
    &\iff \beta G^{1-\gamma}\le 1. \label{eq:foc_poor}
\end{align}
Since $c_t\sim G^t$ and $P_t\sim G^t$ as $t\to\infty$, we obtain the transversality condition for optimality $\beta^t u'(c_t)P_t\sim (\beta G^{1-\gamma})^t\to 0$ if and only if $\beta G^{1-\gamma}<1$, in which case \eqref{eq:foc_poor} holds. Therefore we obtain the following proposition.

\begin{prop}\label{prop:bewley}
If $\beta G^{1-\gamma}<1$ and $b<(\beta G^{1-\gamma})^{1/\gamma}a$, the economy has a bubbly equilibrium with asset price $P_t=pG^t$ for $p$ in \eqref{eq:foc_rich}.
\end{prop}

Many papers build on this example, including \citet{Woodford1990}, \citet[Example 1]{Kocherlakota1992}, \citet[Example 7.1]{HuangWerner2000}, and \citet[Example 1]{Werner2014}.

\citet{ScheinkmanWeiss1986} extend \citet{Bewley1980}'s model to continuous time with endogenous labor supply. Specifically, there are two agents indexed by $i=1,2$ as well as two exogenous aggregate states indexed by $i=1,2$. States exogenously switch at Poisson rate $\lambda>0$, and in state $i$, only agent $i$ can supply labor. There is a linear production technology that produces the consumption good from labor one-for-one. As in \citet{Bewley1980}, there is a pure bubble asset in unit supply, which cannot be shorted. Each agent $i$ seeks to maximize utility
\begin{equation*}
    \E_0\int_0^\infty \e^{-\beta t}(u(c_{it})-l_{it})\diff t
\end{equation*}
subject to the budget and shortsales constraints, where $\beta>0$ is the discount rate, $u$ is the flow utility from consumption $c_{it}\ge 0$, and $l_{it}\ge 0$ is labor supply. Note that the utility function is quasi-linear in labor. In this model, the state variables are the asset holdings of agent 1 denoted by $x\in [0,1]$ and the productivity state $i\in\set{1,2}$. \citet{ScheinkmanWeiss1986} prove the existence of a recursive equilibrium with asset price $P(x,i)>0$ and derive a differential equation for $P(x,i)$ when $u(c)=\log c$. Because the productivity state is persistent, in equilibrium the unproductive agent gradually sells the asset to finance consumption. \citet{ScheinkmanWeiss1986} find that this model can generate rich aggregate dynamics.

\subsection{Sufficient conditions for bubbles}\label{subsec:found_sufficient}

The pure bubble models of \citet{Samuelson1958} and \citet{Bewley1980} are specific examples, and it is of theoretical interest to know under what general conditions bubbles are possible. \citet{OkunoZilcha1983} study a stationary OLG model, meaning endowments and preferences are time-invariant. They show the existence of a Pareto efficient steady state and that the steady state allocation is a competitive equilibrium with or without valued fiat money if either the number of goods at each time or the number of agent types in each generation equals 1. Although they do not explicitly mention it, an immediate corollary is that in a one-good, one-agent OLG model like \S\ref{subsec:found_OLG}, if the autarky allocation is Pareto inefficient, then there exists a bubbly equilibrium, which is Pareto efficient.

\citet{AiyagariPeled1991} extend the model of \citet{OkunoZilcha1983} to a Markov setting (with potentially a linear storage technology) and show that a stationary allocation is Pareto efficient if and only if the matrix of Arrow prices has spectral radius at most 1. Furthermore, they prove the existence of such equilibria. The proof strategy is to let the asset pay dividend $\epsilon D(s)$ in state $s$, prove the existence of efficient stationary equilibrium, and let $\epsilon\to 0$. To illustrate this result, consider the model in \S\ref{subsec:found_OLG}. At the autarky allocation, the Arrow price (the price of a risk-free asset that pays 1 next period) equals
\begin{equation*}
    q=(U_z/U_y)(a,b)=\frac{\beta a}{(1-\beta)b}.
\end{equation*}
Hence Pareto inefficiency (which implies the existence of a bubbly equilibrium) is equivalent to $q>1$, which is precisely the condition in Proposition \ref{prop:samuelson}\ref{item:samuelson_b}. \citet{BarbieHillebrand2018} consider the case with a general concave production function as in \citet{Tirole1985} but with \iid productivity shocks and obtain necessary and sufficient conditions for the existence of bubbly Markov equilibria. \citet{BloiseCitanna2019} study an infinite-horizon incomplete-market model with limited enforcement of debt contracts and show the existence of a bubbly equilibrium when there are gains from trade (autarky is conditionally Pareto inefficient).

\subsection{Necessary conditions for bubbles}\label{subsec:found_necessary}

We next turn to necessary conditions for existence of bubbles, or equivalently sufficient conditions for nonexistence of bubbles. For illustration, consider the infinite-horizon, deterministic setting in \S\ref{sec:prelim} with sequential trade, and assume shortsales are not allowed. Consider an infinitely-lived agent with utility $\sum_{t=0}^\infty \beta^tu(c_t)$, where $\beta>0$ is the discount factor and $u$ is the period utility function satisfying $u'>0$. If the shortsales constraint does not bind, then the first-order condition for optimality implies the Euler equation
\begin{equation}
    u'(c_t)P_t=\beta u'(c_{t+1})(P_{t+1}+D_{t+1}). \label{eq:Euler}
\end{equation}
Comparing the Euler equation \eqref{eq:Euler} to the no-arbitrage condition \eqref{eq:noarbitrage}, the Arrow-Debreu price must be $q_t=\beta^t u'(c_t)/u'(c_0)$ and hence the fundamental value of the asset is
\begin{equation}
    V_0\coloneqq \sum_{t=1}^\infty \beta^t\frac{u'(c_t)}{u'(c_0)}D_t. \label{eq:V0}
\end{equation}
By assumption, the shortsales constraint does not bind. Suppose in addition that the shortsales constraint is uniformly slack, meaning that the agent can reduce the asset holdings by small enough $\epsilon>0$ in every period. Now consider the following feasible trading strategy: starting from the optimal consumption plan $\set{c_t}$, the agent sells $\epsilon>0$ shares of the asset at $t=0$ and keeps this position forever. Then the agent has additional income $P_0\epsilon$ at $t=0$ from the proceeds of sales but gives up dividend income $D_t\epsilon$ at $t>0$. The change in lifetime utility, to first-order approximation, is
\begin{align*}
    0&\ge \underbrace{u(c_0+P_0\epsilon)+\sum_{t=1}^\infty \beta^tu(c_t-D_t\epsilon)}_{\text{feasible}}-\underbrace{\sum_{t=0}^\infty \beta^tu(c_t)}_{\text{optimal}}\\
    &\approx u'(c_0)P_0\epsilon-\sum_{t=1}^\infty \beta^tu'(c_t)D_t\epsilon=\epsilon u'(c_0)(P_0-V_0),
\end{align*}
where the last line uses \eqref{eq:V0}. Therefore $P_0\le V_0$, and because $P_0\ge V_0$ always, we must have $P_0=V_0$: there is no bubble. The intuition is straightforward: if $P_0>V_0$, an unconstrained agent can sell the overpriced asset to enjoy an abnormally high consumption today at the expense of reducing future consumption by giving up dividends, which increases lifetime utility. Thus, in order for an asset price bubble to exist, agents must be financially constrained. In other words, agents' asset holdings must approach 0 infinitely often, which is the essence of \citet[Proposition 3]{Kocherlakota1992}. As discussed in \citet{Kamihigashi1998}, this argument fails in OLG models because the old liquidate asset holdings before exiting the economy, so reducing asset holdings by $\epsilon$ forever is infeasible. In contrast, in representative-agent models, because the agent holds the entire asset in equilibrium, it is very difficult to generate bubbles \citep{Kamihigashi1998,MontrucchioPrivileggi2001}.

In a bubbly equilibrium, by the above argument the asset holdings of every agent need to approach 0 infinitely often. But if the aggregate supply of the asset is positive, because the asset needs to be held by some agent, there must exist an agent whose asset holdings fluctuate between two positive numbers infinitely often. \citet[Proposition 4]{Kocherlakota1992} shows that the present value of the endowments of such an agent is infinite, and consequently, so is the present value of the aggregate endowment.\footnote{The proof contained an error and was corrected by \citet{KocherlakotaToda2023JET}.} We illustrate this point using an argument based on \citet{SantosWoodford1997}. Consider an infinite-horizon economy with time indexed by $t=0,1,\dotsc$ and agents indexed by $i\in I$, where $I$ is either a finite or countably infinite set. Following the notation in \S\ref{sec:prelim}, let $q_t>0$ be the Arrow-Debreu price, $q=(q_t)_{t=0}^\infty$ the price vector, $P_t$ be the asset price, and $D=(D_t)_{t=0}^\infty$ be the dividend stream. Let $c_i=(c_{it})_{t=0}^\infty$ and $e_i=(e_{it})_{t=0}^\infty$ be the consumption and endowment vectors of agent $i$. Suppose agent $i$ is endowed with $x_{i0}$ shares of the asset, where $\sum_{i\in I}x_{i0}=1$.

With locally nonsatiated preferences, the budget constraint implies
\begin{equation*}
    q\cdot c_i=q\cdot e_i+(P_0+D_0)x_{i0}.
\end{equation*}
Aggregating across agents and using $\sum_{i\in I}x_{i0}=1$, we obtain
\begin{equation*}
    q\cdot \sum_{i\in I} c_i=q\cdot \sum_{i\in I}e_i+P_0+D_0.
\end{equation*}
Using the commodity market clearing condition $\sum_{i\in I}c_i=\sum_{i\in I}e_i+D$ and $q_0=1$, we obtain
\begin{equation*}
    q\cdot \sum_{i\in I}e_i+\sum_{t=1}^\infty q_tD_t=q\cdot \sum_{i\in I}e_i+P_0.
\end{equation*}
Therefore if the present value of aggregate endowment is finite, so $q\cdot \sum_{i\in I}e_i<\infty$, canceling this term from both sides yields $P_0=\sum_{t=1}^\infty q_tD_t\eqqcolon V_0$, so the asset is priced at the fundamental value. The contrapositive is that if there is a bubble, so $P_0>V_0$, then the present value of aggregate endowment must be infinite:
\begin{equation}
    q\cdot \sum_{i\in I}e_i=\infty. \label{eq:infinite_endowment}
\end{equation}
For instance, in the \citet{Samuelson1958} model, if we focus on the stationary bubbly equilibrium, because the asset price $P_t$ is constant, the gross risk-free rate is $R_t=P_{t+1}/P_t=1$. Because the aggregate endowment in each period is $a+b$, the present value of the aggregate endowment is infinite. Similarly, in the \citet{Bewley1980} model, the gross risk-free rate is $R_t=P_{t+1}/P_t=G$, and the aggregate endowment at time $t$ is $(a+b)G^t$, so the present value is again infinite.

\citet[Theorem 3.3]{SantosWoodford1997} significantly extend this result under incomplete markets and borrowing constraints and show that if the present value of the aggregate endowment is finite, then the price of an asset in positive net supply or with finite maturity equals its fundamental value. Furthermore, their Corollary 3.4 (together with Lemma 2.4) shows that, when the asset pays nonnegligible dividends relative to the aggregate endowment, so there exists $\eta>0$ such that $D\ge \eta \sum_{i\in I}e_i$, bubbles are impossible. To see why, if there is a bubble, we know $q\cdot \sum_{i\in I}e_i=\infty$, but then
\begin{equation*}
    P_0+D_0\ge V_0+D_0=q\cdot D\ge \eta q\cdot \sum_{i\in I}e_i=\infty,
\end{equation*}
which is a contradiction. Thus the results of \citet{SantosWoodford1997} have been interpreted as fundamental difficulties for generating asset price bubbles in dividend-paying assets.

The Bubble Impossibility Theorem of \citet{SantosWoodford1997} has been extended in several directions, including debt constraints \citep{Kocherlakota2008,Werner2014} and collateral constraints \citep*{AraujoPascoaTorres-Martinez2011}. Furthermore, \citet{HellwigLorenzoni2009} study a general equilibrium model with limited commitment and show that equilibrium allocations with self-enforcing private debt is equivalent to those sustained with unbacked public debt or rational bubbles.

\section{Further topics on rational bubbles}\label{sec:topic}

\subsection{Stochastic bubbles}

\citet{Blanchard1979} studies a reduced-form asset pricing equation and considers the possibility that a bubble ends with positive probability. \citet{Weil1987} introduces this mechanism of ``stochastic bubble'' into the two-period OLG model. A merit of stochastic bubbles is to examine the macroeconomic impact of bubbles by taking into account the possibility of their bursts. As an illustration, consider the model in \S\ref{subsec:found_OLG}. We seek an equilibrium of the following form:
\begin{enumerate*}
    \item the asset initially trades at price $P>0$,
    \item each period, the asset trades at price $P>0$ with probability $\upsilon$ and becomes worthless (complete collapse) with probability $1-\upsilon$, 
    \item once the asset becomes worthless, its re-emergence is not expected.
\end{enumerate*}
Letting $x$ be the asset holdings of the young, the objective function is
\begin{equation*}
    (1-\beta)\log(a-Px)+\beta(\upsilon\log(b+Px)+(1-\upsilon)\log b).
\end{equation*}
Taking the first-order condition and imposing the market clearing condition $x=1$, we obtain
\begin{equation*}
    -\frac{1-\beta}{a-P}+\frac{\upsilon\beta}{b+P}=0\iff P=\frac{\upsilon\beta a-(1-\beta)b}{1-\beta+\upsilon\beta}.
\end{equation*}
This is obviously a rational expectations equilibrium if $\upsilon\beta a>(1-\beta)b$. The case $\upsilon=1$ corresponds to the deterministic model in \S\ref{subsec:found_OLG}. The condition implies that bubbles with a high probability of collapse cannot exist. Intuitively, in order to compensate for the risk of bursting, those bubbles would need to grow so fast that they would become too large to be feasible (to be sustained by the young's endowment) in equilibrium. By a similar argument to Proposition \ref{prop:samuelson}, there also exist a continuum of nonstationary equilibria in which the asset price converges to zero conditional on the bubble not bursting. There are many applications of stochastic bubbles, including \citet{CaballeroKrishnamurthy2006}, \citet{Kocherlakota2009}, \citet{HiranoYanagawa2010,HiranoYanagawa2017}, \citet{FarhiTirole2012}, \citet*{HiranoInabaYanagawa2015}, \citet{Clain-Chamosset-YvrardKamihigashi2017}, \citet*{BiswasHansonPhan2020}, \citet{Bonchi2023}, and \citet{HoriIm2023insurance}. 

\subsection{Crowd-out and crowd-in effects of bubbles}

An asset price bubble diverts savings from capital to the bubble asset, and \emph{ceteris paribus}, reduces capital accumulation and potentially growth. To illustrate this crowd-out effect of bubbles, we consider \citet{Tirole1985}'s OLG model with capital accumulation. For simplicity, suppose that each generation has the Cobb-Douglas utility function \eqref{eq:CD}, the young have one unit of labor endowment that is inelastically supplied, the old have no labor endowment, and the aggregate production function is Cobb-Douglas with capital depreciation $\delta\in [0,1]$:
\begin{equation}
    F(K,L)=AK^\alpha L^{1-\alpha}+(1-\delta)K. \label{eq:F}
\end{equation}
Let us focus on a steady state. Profit maximization implies the wage and rental rate (gross return on capital)
\begin{subequations}
\begin{align}
    w&=F_L(K,1)=A(1-\alpha)K^\alpha, \label{eq:w_Tirole}\\
    R&=F_K(K,1)=A\alpha K^{\alpha-1}+1-\delta. \label{eq:R_Tirole}
\end{align}
\end{subequations}
Due to Cobb-Douglas utility, aggregate savings equals $\beta w$. In a fundamental steady state, aggregate savings equals capital, so
\begin{equation}
    K=\beta A(1-\alpha)K^\alpha\iff K_f\coloneqq [\beta A(1-\alpha)]^\frac{1}{1-\alpha}. \label{eq:Kf_Tirole}
\end{equation}
Using \eqref{eq:R_Tirole}, we obtain the fundamental gross risk-free rate
\begin{equation*}
    R_f\coloneqq A\alpha K_f^{\alpha-1}+1-\delta=\frac{1}{\beta}\frac{\alpha}{1-\alpha}+1-\delta.
\end{equation*}
We next consider a bubbly steady state. If a bubbly steady state exists, because the gross return on the pure bubble asset is 1, we must have
\begin{equation}
    1=R_b=A\alpha K^{\alpha-1}+1-\delta\iff K_b\coloneqq (A\alpha/\delta)^\frac{1}{1-\alpha}. \label{eq:Kb_Tirole}
\end{equation}
Aggregate savings equals capital plus the bubble price, so \eqref{eq:Kf_Tirole} is modified to
\begin{equation*}
    K+P=\beta A(1-\alpha)K^\alpha\iff P=(A\alpha/\delta)^\frac{1}{1-\alpha}\left(\beta\delta\frac{1-\alpha}{\alpha}-1\right).
\end{equation*}
In order for $P>0$, it is necessary and sufficient that
\begin{equation}
    \beta\delta\frac{1-\alpha}{\alpha}-1>0\iff R_f=\frac{1}{\beta}\frac{\alpha}{1-\alpha}+1-\delta<1. \label{eq:bubble_cond_Tirole}
\end{equation}

Therefore for parameter specifications under which a bubbly steady state exists, we have $R_f<R_b=1$ and hence $K_f>K_b$. The intuition is straightforward. If there is no pure bubble asset, all savings by the young flows to capital investment. But once pure bubble assets pop up in the economy, they crowd young's savings away from capital investment, thereby reducing the capital stock. \citet{Saint-Paul1992} extends \citet{Blanchard1985JPE}'s perpetual youth model with endogenous growth and shows that an increase in public debt reduces long-run economic growth. \citet{GrossmandYanagawa1993} obtain a similar result in a two-period OLG pure bubble model. \citet{CahucChalle2012} show the mechanism of how asset bubbles distort the allocation of labor across sectors, \ie, bubbles crowd out productive labor, in addition to crowding out productive capital. \citet{Plantin2023} shows that an accommodative monetary policy has a crowd-out effect on productive investments by creating asset bubbles.  

The crowd-out effect also arises in infinite-horizon models. \citet*{AokiNakajimaNikolov2014} study a model with infinitely-lived agents with log utility who have access to a stochastic storage technology ($AK$ model) subject to idiosyncratic investment risk, which builds on the OLG model of \citet{Kitagawa1994}. They show that asset bubbles reduce economic growth through the crowd-out effect. \citet*{Guerron-QuintanaHiranoJinnai2023} construct a model of recurrent bubbles and show that expectations about future bubbles have a crowd-out effect on capital investments even before realization. Hence, high-frequency bubbles reduce average economic growth.

Bubbly episodes (such as the U.S. dot-com bubble) are typically associated with a high level of investment, which goes against the crowd-out effect. \citet{Woodford1990} introduces capital accumulation in \citet{Bewley1980}'s model discussed in \S\ref{subsec:found_infhor}. In one of his models, agents differ not by endowments but by investment opportunities: an investment opportunity arrives to each agent every other period. Borrowing and lending is impossible. In this environment, government bonds (which play the same role as bubble assets) crowd investment in because agents without investment opportunities purchase government bonds as a means of savings and sell when they have investment opportunities, which finance more investment.

To illustrate the crowd-in effect, consider the following modification to \citet{Tirole1985}'s model. A newborn agent has entrepreneurial ability with probability $\pi$: specifically, the agent has access to the production technology \eqref{eq:F} with probability $\pi$, while $A=0$ with probability $1-\pi$, which is realized before investment. Agents cannot borrow or lend but can purchase a pure bubble asset. In a fundamental steady state, aggregate savings $\beta w$ equals aggregate capital $K$. By assumption, fraction $\pi$ of capital is productive, while fraction $1-\pi$ goes to storage. Therefore \eqref{eq:Kf_Tirole} becomes
\begin{equation}
    \pi K+(1-\pi)K=\beta A(1-\alpha)(\pi K)^\alpha\iff K_f\coloneqq [\beta A(1-\alpha)\pi^\alpha]^\frac{1}{1-\alpha}. \label{eq:Kf_crowdin}
\end{equation}
In contrast, in the bubbly steady state, unproductive agents hold the pure bubble asset (which yields a gross return 1, higher than $1-\delta$ from storage) and productive agents hold aggregate capital $K$, so \eqref{eq:Kb_Tirole} continues to hold. Therefore if $\pi$ is sufficiently small, then \eqref{eq:Kb_Tirole} and \eqref{eq:Kf_crowdin} imply $K_b>K_f$: we have crowd-in.

The crowd-in effect could arise for various mechanisms but in most papers written since the 2008 financial crisis, the idea is based on \citet{Woodford1990}'s mechanism of ex-ante heterogeneity, including \citet{Kocherlakota2009}, \citet{HiranoYanagawa2010, HiranoYanagawa2017}, \citet*{HiranoInabaYanagawa2015}, \citet*{BiswasHansonPhan2020} (infinite-horizon models with heterogeneous productivity), or \citet{FarhiTirole2012} and \citet*{ClainRaurichSeegmuller2023} (three-period OLG models with heterogeneous savings motive). These papers consider stochastic bubbles instead of safe government bonds. Since bubble assets are risky, they yield higher returns than safe assets, which endogenously increases the wealth of productive entrepreneurs with investment opportunities through general equilibrium effects, generating crowd-in effects. The mechanism in \citet{MartinVentura2012} is different. They employ a two-period OLG model assuming that the young with investment opportunities cannot borrow but can create bubble assets (money creation). This assumption of liquidity creation directly produces wealth effects for productive young agents and is crucial for their crowd-in effect. \citet{Queiros2023} develops a two-period OLG model and studies the interplay between bubble creations by newborn firms and product market competition. He shows that bubble creation can provide an entry subsidy, producing crowd-in effects on capital. \citet{Bonchi2023} considers a three-period OLG model in which the middle-aged can create bubbles and the young use future bubbles as a collateral, which produces a crowd-in effect.

\subsection{Bubbles and financial conditions}

Bubbles are related to low interest rate environments. At the fundamental equilibrium allocation in the OLG model of \S\ref{subsec:found_OLG}, the gross risk-free rate is
\begin{equation*}
    R=(U_y/U_z)(a,b)=\frac{(1-\beta)b}{\beta a}<1
\end{equation*}
under the bubble existence condition in Proposition \ref{prop:samuelson}\ref{item:samuelson_b}. The same is true in \citet{Tirole1985}'s model: see \eqref{eq:bubble_cond_Tirole}. Thus, equilibria with pure bubbles can exist when the fundamental interest rate is negative (lower than the growth rate of the economy), implying that agents have a strong savings motive. This result also holds in the presence of deterministic or stochastic storage technologies as in \citet{Wallace1980} and \citet{AiyagariPeled1991}.

As discussed in \S\ref{subsec:found_necessary}, asset price bubbles can arise only if the present value of the aggregate endowment is infinite. In OLG models, this is not so demanding because infinite present value and infinitely many agents with finite lives could be mutually consistent. In contrast, in infinite-horizon models, some financial frictions such as borrowing constraints, shortsales constraints, or imperfect insurance markets are necessary to prevent agents from capitalizing infinite present value. In this respect, infinite-horizon models illustrate the role of financial frictions in generating asset bubbles. The literature since the 2008 financial crisis has mainly examined the relationship between asset bubbles and financial conditions in various settings. In the following, we focus on infinite-horizon models. 

\citet{Kocherlakota2009} considers an infinite-horizon model in which agents have log utility and have access to a Cobb-Douglas production function (investment opportunities arrive) every period with probability $\pi\in (0,1)$, which is a special case of the model of \citet{Angeletos2007}. In this economy, productive agents are willing to borrow from unproductive agents. However, \citet{Kocherlakota2009} assumes that borrowing must be collateralized by a durable asset (``housing''). In equilibrium, even though housing is intrinsically useless and costly to produce, it can be used as collateral and have a positive price, thereby facilitating lending and borrowing, as in \citet{Bewley1980} and \citet{Woodford1990}. \S\ref{sec:bare-bones} discusses such a model in detail. 

\citet{HiranoYanagawa2010,HiranoYanagawa2017} consider an $AK$ model with infinitely-lived heterogeneous agents with different productivities, as in \citet{Kiyotaki1998}. Productive agents wish to borrow, but borrowing cannot exceed a constant multiple of output. They find that bubbles are likely to emerge when the degree of pledgeability is in the middle range. 
They also connect the pledgeability constraint to the size of the crowd-in and crowd-out effects and find that within the bubble region, the crowd-in effect dominates and boosts economic growth when pledgeability is relatively low. \citet{HoriIm2023insurance} introduce imperfect insurance markets into a growth model with infinitely-lived agents. Entrepreneurs face uninsurable entrepreneurial risks but no credit constraints. They show that if the degree of entrepreneurial risks is in the middle range, bubbles are likely to emerge.\footnote{\citet{MiaoWang2018} discuss stock price bubbles in a steady-state equilibrium in relation to financial conditions. In their model, the stock price equals $V=QK+B$, where $K$ is capital, $Q$ is the price of capital, and $B$ is a constant interpreted as the bubble component. Several papers study similar models, including \citet*{MiaoWang2012, MiaoWang2014,MiaoWang2015}, \citet*{MiaoWangXu2015}, \citet*{MiaoWangXu2016}, and \citet{Ikeda2022}. All of these papers appear to have been understood as models of rational asset price bubbles attached to real assets. However, we know from Lemma \ref{lem:bubble} that bubbles in dividend-paying assets can never occur in a steady-state equilibrium. An appropriate interpretation would be that there are two steady states, one with high stock prices and the other with low stock prices. In both steady states, stock prices always reflect fundamentals, but expectations determine which steady-state equilibrium is reached. Obviously, the literature that studies the effects of self-fulfilling expectations on macroeconomic outcomes including asset prices is important, as \citet{Azariadis1981}, \citet{CassShell1983}, and \citet{Farmer1999} have paved the way. Our paper focuses on rational asset price bubbles but the two approaches are complementary and provide different insights for the determination of asset prices.}

There are also papers that analyze the effect of financial conditions on the aftermath of the burst of bubbles. For instance, \citet{HiranoYanagawa2010,HiranoYanagawa2017} show that the growth path after a collapse of bubbles qualitatively depends on financial conditions. When the condition is good, it produces V-shape recovery, while it is not, it produces a L-shape and permanent decline in economic growth. \citet*{BiswasHansonPhan2020} extend \citet*{HiranoInabaYanagawa2015} to include zero lower bound and nominal wage rigidity constraints. They show that a bursting bubble can push the economy into a secular stagnation equilibrium, where both constraints simultaneously bind, leading to a persistent recession.

\subsection{Monetary and fiscal policy}

Given the results of crowd-in and crowd-out effects of pure bubbles and the relationship between financial conditions and pure bubbles, many papers apply those results to monetary and fiscal policy. 

Concerning monetary policy and asset bubbles, a main question is whether monetary policy should react to asset price bubbles. \citet{Clain-Chamosset-YvrardSeegmuller2015} show that a monetary policy rule that responds only to inflation is destabilizing and promotes local indeterminacy, but a rule that responds also to asset bubbles can be globally and locally stabilizing. \citet{Gali2014} employs a two-period OLG model and shows that raising interest rates may end up increasing the size of bubbles, instead of decreasing it, while \citet*{DongMiaoWang2020} show that it reduces bubble volatility in an infinite-horizon economy, and \citet*{AllenBarlevyGale2023} show in a two-period OLG model that raising interest rates sufficiently eliminates any pure bubble equilibria. \citet*{AsriyanFornaroMartinVentura2021} consider optimal monetary policy in a two-period OLG model where the central bank creates money and private agents also create another money and make seigniorage revenues.

Concerning fiscal policy, a main question is what kind of fiscal policy is desirable after the burst of asset bubbles, or in light of the possibility of a bubble collapse. For instance, \citet{Kocherlakota2009} shows that issuing government bonds backed by taxes after a collapse of asset bubbles restores efficiency. \citet*{HiranoInabaYanagawa2015} extend \citet{HiranoYanagawa2010, HiranoYanagawa2017} to include government bailouts. They take into account the possibility of a bubble burst and derive the optimal time-invariant bailout policy for entrepreneurs who suffer losses by bubble bursts and workers who are taxpayers, respectively. In addition, unlike \citet{Kocherlakota2009}, they also show risky bubbles that are expected to collapse can be better than deterministic bubbles (or government bonds) from the welfare perspective of workers. \citet{AokiNikolov2015} also develop a related model of \citet{HiranoYanagawa2010, HiranoYanagawa2017} to analyze bank bailouts quantitatively. \citet{DongXu2022} is related to \citet*{HiranoInabaYanagawa2015} but examine the efficiency and welfare implications of the time-varying bailout policy. \citet{BosiPham2016} show that taxing the bubble asset could overturn the crowd-out effect in the \citet{GrossmandYanagawa1993} model.

\subsection{Bubbles in open economies}

There are also some applied papers of pure bubble models to open economies. The main questions are the effect of bubbles and their collapse on capital flow and the relationship between globalization and asset bubbles in open economies. For instance, \citet{CaballeroKrishnamurthy2006} construct a two-period OLG model of a small open economy and analyze movements in capital flow in emerging markets that occur with the onset and collapse of bubbles. \citet{Motohashi2016} examines how capital account liberalization affects the existence of bubbles in an infinite-horizon small open economy. Regarding two-country models, \citet{Basco2014}, \citet{MartinVentura2015}, \citet{Rondina2017}, and \citet{IkedaPhan2019} study how financial integration affects the existence of bubbles in two-period OLG models. \citet{Clain-Chamosset-YvrardKamihigashi2017} show in a two-country OLG exchange economy that a bubble crash in the foreign country has contagious effects, inevitably leading to a bubble crash in the home country. \citet{Shimizu2018} develops an infinite-horizon two-country model and studies the effects of asset bubbles and their collapse on each country’s economic growth.  

\subsection{Housing bubbles}

Based on the motivation that many countries have experienced housing price booms and busts, there are also papers on housing bubbles. \citet{ArceLopez-Salido2011} consider a three-period OLG model in which the young derive utility from housing service. Their model features two market imperfections: a down payment (collateral) constraint for housing purchase and the absence of a rental market. They show in Proposition 3 that when the collateral constraint is tighter than a threshold, the interest rate becomes low and the middle-aged hold vacant housing (a pure bubble asset) as a means of savings. The model of \citet{Zhao2015} is nearly identical except that it is a two-period OLG model and the middle-aged is replaced with investors who do not derive utility from housing. \citet{ChenWen2017} extend the model of \citet{Tirole1985} to two sectors with heterogeneous productivity (but identical technological growth rates) and study the transition dynamics. \citet*{JiangMiaoZhang2022} consider a two-sector OLG model with a government, where housing is a pure bubble asset produced by labor and land (owned by the government) and the government provides infrastructure that exhibits positive externality on production. Although these models may capture some aspects of housing bubbles, housing does not generate rents and its fundamental value is zero. Thus, these models are the same as pure bubble models.

\subsection{Criticisms to pure bubble models}\label{subsec:topic_criticism}

As we have seen so far, the seminal papers by \citet{Samuelson1958} and \citet{Tirole1985} have produced a vast literature and fruitful outcomes. However, the existing pure bubble literature has severe limitations. There are three major criticisms. First, although pure bubble models are useful to analyze money or monetary economies, in reality it is hard to find pure bubble assets other than money or cryptocurrency. Hence, pure bubble models are unrealistic for describing bubbles attached to real assets such as stocks, housing, and land. Indeed, by definition, they exclude dividend-paying assets so the scope for empirical or quantitative analysis to think about realistic bubbles is severely limited.

Second, pure bubble models suffer from equilibrium indeterminacy as in Proposition \ref{prop:samuelson}, which makes model predictions non-robust. Although the equilibrium indeterminacy in pure bubble models has been recognized at least since \citet{Gale1973}, the literature has selected only one of a continuum of bubbly equilibria (a saddle path or a steady state) and has advanced policy and quantitative analysis. However, this selection is not necessarily convincing because the equilibrium selected is only one point within an open set and thus has measure zero. Indeed, it is known that monetary policy implications much depend on which equilibrium we focus on \citep*{Gali2014,DongMiaoWang2020}. 

Third, there is an econometric literature on detecting asset price bubbles attached to real assets by using the price-dividend ratio \citep*{PhillipsShiYu2015,PhillipsShi2018,PhillipsShi2020}, which is considered important for policy considerations in the real world. However, pure bubble models are completely silent about the price-dividend ratio because it cannot be even defined. Hence, it is impossible to bridge the gap between pure bubble models and the bubble detection literature and, therefore, to derive meaningful policy implications from a macroeconomic perspective based on the econometric literature.

These three criticisms simply show that in describing bubbles attached to real assets, the pure bubble theory faces a fundamental limitation for applications. If the theory cannot be applied, it will be hard for the literature to develop. Indeed, although the theory of pure bubbles and financial conditions have developed since the 2008 financial crisis, there are only a few papers that apply pure bubble models to modern empirical and quantitative analysis \citep*{DomeijEllingsen2018,Guerron-QuintanaHiranoJinnai2023}.\footnote{When one of the authors (Hirano) presented earlier papers on pure bubbles \citep*{HiranoYanagawa2010,HiranoYanagawa2017,HiranoInabaYanagawa2015}, the general audience often commented that because of these criticisms, pure bubble models are useless for thinking about stock, land, and housing price bubbles and proceeding with policy analyses. Some researchers told us personally that they left the literature due to these shortcomings and difficulties in generating bubbles attached to real assets.}  

\section{Necessity of bubbles}

Following \citet{Samuelson1958}, the literature on rational bubbles has almost exclusively focused on pure bubbles, which have severe limitations as discussed in \S\ref{subsec:topic_criticism}. This section discusses models of bubbles attached to dividend-paying assets. 

\subsection{Contribution of \texorpdfstring{\citet{Wilson1981}}{}}

The first paper that we are aware of that constructs bubbles attached to dividend-paying assets is \citet{Wilson1981}. His main result is to establish the existence of competitive equilibria when both the number of agents and commodities can be infinite, which in particular includes the classical OLG model. \citet[Theorem 1]{Wilson1981} proves the existence of an equilibrium with transfer payments where the transfers can be set to zero (so budgets balance exactly) for agents endowed with only finitely many commodities. However, transfers need not be zero for agents endowed with infinitely many commodities, for instance a Lucas tree that produces dividends indefinitely.

Here we present a minimal example in the spirit of \citet[\S7]{Wilson1981}. We consider the two-period OLG model discussed in \S\ref{subsec:found_OLG} but with the following modifications: at time $t$,
\begin{enumerate*}
    \item the young has endowment $a_t>0$,
    \item the old has no endowment ($b_t=0$), and
    \item the asset pays dividend $D_t>0$.
\end{enumerate*}
Then we need to modify the budget constraints \eqref{eq:budget_OLG} such that $(a,b)$ and $P_{t+1}$ are replaced with $(a_t,0)$ and $P_{t+1}+D_{t+1}$, respectively. Note that because $D_t>0$, the asset price must be positive: $P_t>0$. The consumption of the young \eqref{eq:yt_OLG} then becomes $y_t=(1-\beta)a_t$, and the asset market clearing condition $x_t=1$ implies the asset price
\begin{equation*}
    P_t=P_tx_t=a_t-y_t=\beta a_t.
\end{equation*}
Therefore we obtain the following proposition.

\begin{prop}\label{prop:wilson}
The two-period OLG model with Cobb-Douglas utility and long-lived asset has a unique equilibrium. The equilibrium asset price is $P_t=\beta a_t>0$ and the dividend yield is $D_t/P_t=D_t/(\beta a_t)$. The equilibrium features an asset price bubble if and only if
\begin{equation}
    \sum_{t=1}^\infty \frac{D_t}{a_t}<\infty. \label{eq:bubble_cond_wilson}
\end{equation}
\end{prop}

\begin{proof}
The uniqueness and characterization of equilibrium follow from the preceding discussion. Since the dividend yield is $D_t/P_t=D_t/(\beta a_t)$, by the Bubble Characterization Lemma \ref{lem:bubble}, the condition \eqref{eq:bubble_cond_wilson} is necessary and sufficient for an asset price bubble.
\end{proof}

Intuitively, pulled up by young's growing incomes (endowments), the asset prices will rise at a faster rate than dividends, deviating from the fundamental value. The example in \citet[\S7]{Wilson1981} differs from Proposition \ref{prop:wilson} in that the utility function is linear and endowments of the old are positive, but it shares the feature that dividends become asymptotically negligible compared to endowments in the sense of \eqref{eq:bubble_cond_wilson}.

\subsection{Contribution of \texorpdfstring{\citet{Tirole1985}}{}}

Proposition 1(c) of \citet{Tirole1985} recognizes the possibility that bubbles are necessary for equilibrium existence if the interest rate without bubbles is negative. Although he gives some explanations on p.~1506 in the sentence starting with ``The intuition behind this fact roughly runs as follows,'' he did not necessarily provide a formal proof. In \citet{Tirole1985}, the proof of the nonexistence of fundamental equilibria appears at the bottom of p.~1522 and the top of p.~1523. The proof uses a convergence result discussed in Lemma 2. However, this convergence heavily relies on the monotonicity condition on the function $\psi$ defined in Equation (7) on p.~1502. This monotonicity/stability condition is a high-level assumption that need not be satisfied in a general setting. In fact, Tirole does not provide any example for which this assumption is satisfied, and in OLG models it is well known that there are robust examples of equilibrium with cycles \citep[\S 5]{GeanakoplosPolemarchakis1991}, which necessarily violates this assumption. To the best of our understanding, his Proposition 1(c) is limited to providing an important conjecture.

\subsection{Bubble Necessity Theorem}
While \citet{Wilson1981} made the surprising discovery of the nonexistence of fundamental equilibria, he only gave an example in a fairly limited setting. Therefore, it is not obvious to what extent there is generality and how relevant the result is, nor is it obvious what new insights and asset pricing implications can be drawn when we consider more general macroeconomic models. 

\citet{HiranoTodaNecessity} generalize \citet{Wilson1981}'s result in workhorse macroeconomic models. They establish the Bubble Necessity Theorem and present a conceptually new idea of the necessity of asset price bubbles. Necessity (bubbles \emph{must} arise under some conditions) is a fundamentally different concept from possibility (bubbles \emph{can} arise in pure bubble models). More specifically, \citet{HiranoTodaNecessity} present a large class of plausible economic models with dividend-paying assets in which asset price bubbles are a necessity or inevitable in the sense that
\begin{enumerate*}
    \item the economy admits a bubbly equilibrium but
    \item neither fundamental equilibria nor bubbly equilibria that become asymptotically bubbleless exist, unlike Proposition \ref{prop:samuelson}.
\end{enumerate*}
In other words, all equilibria are bubbly with non-negligible bubble sizes relative to the economy. The key bubble necessity condition is
\begin{equation}
    R<G_d<G, \label{eq:necessity}
\end{equation}
where $G$ is the long-run endowment growth rate, $G_d$ is the long-run dividend growth rate, and $R$ is the counterfactual long-run autarky interest rate. The intuition is as follows. If a fundamental equilibrium exists, in the long run the asset price (the present value of dividends) must grow at the same rate of $G_d$. Then the asset price becomes negligible relative to endowments because $G_d<G$ and the equilibrium consumption allocation approaches autarky. With an autarky interest rate $R<G_d$, the present value of dividends (the fundamental value of the asset) becomes infinite, which is of course impossible in equilibrium.

\citet{HiranoTodaNecessity} establish the Bubble Necessity Theorem in abstract OLG models and in a Bewley-type model satisfying \eqref{eq:necessity}. The following proposition is a special case. (We intentionally impose strong assumptions for concise statement; see \citet{HiranoTodaNecessity} for more general results.)

\begin{prop}[Necessity of bubbles]\label{prop:necessity}
Consider a two-period OLG model with a long-lived asset satisfying the following conditions.
\begin{enumerate}
    \item The utility function of generation $t$ is $U(y_t,z_{t+1})$, where $U$ is continuously differentiable, homogeneous of degree 1, and quasi-concave.
    \item The date $t$ endowments of the young and old are $(a_t,b_t)=(aG^t,bG^t)$.
    \item The date $t$ dividend of the asset is $D_t=DG_d^t$.
\end{enumerate}
Let $R=(U_y/U_z)(a,b)$ be the long-run autarky interest rate. If \eqref{eq:necessity} holds, then the equilibrium asset price exhibits a bubble.
\end{prop}

Note that Proposition \ref{prop:necessity} does not rule out the possibility of multiple equilibria.\footnote{In addition to showing the necessity of asset bubbles, \citet{HiranoTodaNecessity} also show that under some parameter values, there is a unique asymptotically bubbly equilibrium.} However, and more importantly, it rules out the existence of fundamental equilibria and of bubbly equilibria that become asymptotically bubbleless. Hence, it establishes the necessity of bubbly equilibria.

\citet{HiranoTodaNecessity} also argue that the condition $G_d<G$ naturally arises. In reality, the economy consists of multiple sectors, and there is no reason to expect homogeneous growth rates across sectors. Once we consider heterogeneous growth rates of different sectors or different productivity growth rates of different production factors, \emph{unbalanced growth} occurs and the condition $G_d<G$ naturally holds,\footnote{\citet{Baumol1967} is the seminal paper about unbalanced growth, which led to a literature on structural transformation. For recent developments, see for instance \citet{Matsuyama1992}, \citet{AcemogluGuerrieri2008}, \citet{BueraKaboski2012}, \citet{Boppart2014}, \citet*{CominLashkariMestieri2021}, and \citet{FujiwaraMatsuyama2022}. \citet{HiranoTodaUnbalanced} is a simple example that considers the interactions between structural transformation and asset price bubbles.} implying that asset price bubbles will emerge within the larger historical process of structural transformation. Indeed, looking back at history, it is more natural to assume that when new technological innovations occur, including the Industrial Revolution, they cause unbalanced growth. 

Based on the idea that unbalanced growth and asset price bubbles are closely connected, \citet{HiranoTodaUnbalanced} present a model with unbalanced growth and show that a land price bubble necessarily emerges when we consider the structural transformation from a Malthusian economy where land plays an important role as a factor for production, to a modern economy where the role of land as a factor of production diminishes and technological innovations drive economic growth. \citet{HiranoTodaHousingbubble} study an OLG model with perfect housing and rental markets (so housing is an asset with endogenous dividends, unlike pure bubble models) and show that during the process of economic development with unbalanced growth, including easier access to credit, housing bubbles will necessarily emerge. In both papers, \eqref{eq:necessity} is satisfied. These papers show that land and housing price bubbles inevitably emerge with economic and financial development.

The Bubble Necessity Theorem has two important implications. First, in canonical modern macro-finance models, it seems that there is a presupposition that asset prices \emph{must} reflect fundamentals, and indeed models are often constructed in such a way by imposing stationarity. Contrary to this view, the Bubble Necessity Theorem implies that asset prices \emph{cannot} equal fundamental values under some circumstances. In other words, there are benchmark cases in which the notion that asset prices should reflect fundamentals is false.\footnote{As noted in the main text, the concept of the \emph{necessity} of asset price bubbles is fundamentally different from the \emph{possibility} of bubbles in pure bubble models. In the former case, bubbles are inevitable for the existence of equilibrium. Hence, given the same parameter values, we cannot compare economies with and without bubbles and we need to change the conventional way of thinking about welfare and policy implications of bubbles. \citet{Barlevy2018} discusses welfare and policy implications in pure bubble models.} Second, asset pricing implications under balanced growth and unbalanced growth are markedly different. In the next section, we present a simple model to illustrate this idea.

\section{Bare-bones infinite-horizon model}\label{sec:bare-bones}

So far, most of our discussion has centered around OLG models, which are arguably stylized. The purpose of this section is to present a bare-bones infinite-horizon (Bewley-type) model that generates an asset price bubble attached to a dividend-paying asset. This model illustrates the idea that technological progress and structural transformation from a land economy to a modern economy cause unbalanced growth associated with land price bubbles. The model we present here is a much simplified version of \citet*{HiranoJinnaiTodaLeverage}. 

\subsection{Model}

We consider a two-sector economy with a capital-intensive sector and a real estate sector. There is one homogeneous good and a mass 1 continuum of economic agents.

\paragraph{Preferences}

A typical agent has the logarithmic utility function
\begin{equation}
    \E_0\sum_{t=0}^{\infty}\beta^t\log c_{it}, \label{eq:utility}
\end{equation}
where $\E_t[\cdot]$ denotes the expectation conditional on date $t$ information, $\beta\in(0,1)$ is the subjective discount factor, $i\in I=[0,1]$ indexes agents, and $c_{it}$ is consumption of agent $i$ at date $t$. 

\paragraph{Investment technologies}

At each date $t$, an investment opportunity arrives to each agent with probability $\pi\in (0,1)$ independently across time and agents. We call an agent with (without) an investment opportunity an entrepreneur (saver).

A unit of investment in period $t$ produces a unit of capital in period $t$, which is available for production in period $t+1$. A unit of capital produces $A>0$ units of consumption goods, so $A$ captures productivity of capital. After production, a fraction $\delta\in [0,1]$ of capital depreciates. Hence, the gross return on capital is $A+1-\delta$.

\paragraph{Real estate sector}

A unit of land (real estate) yields $D>0$ units of consumption goods as dividends (land rents) every period. We assume that the aggregate supply of land is fixed at $X>0$. More generally, we can consider a neoclassical (constant-returns-to-scale) production function. For instance, we may let the output take the constant elasticity of substitution form
\begin{equation*}
    Y_t=\left(\alpha(A_t^{K}K_t)^{1-1/\sigma}+(1-\alpha)(A_t^{X}X)^{1-1/\sigma}\right)^{\frac{1}{1-1/\sigma}},
\end{equation*}
where $\alpha\in (0,1)$ is a parameter, $\sigma\in (0,\infty]$ is the elasticity of substitution, and $A_t^K$ and $A_t^X$ are productivity of capital and land at date $t$. The bare-bones model is the special case of $\sigma=\infty$ without productivity growth, \ie, $A_t^K=A^K$ and $A_t^X=A^X$ with $A\coloneqq \alpha A^K$ and $D\coloneqq (1-\alpha)A^X$. The case $\sigma=\infty$ is convenient because it gives closed-form solutions and clearly illustrates the mechanism that generates land price bubbles.\footnote{\citet*{HiranoJinnaiTodaLeverage} show that the qualitative results are robust to the specification of the production function.}

\paragraph{Budget constraint}

We suppose that agents cannot borrow and hence self-finance. Let $k_{it},x_{it}\ge 0$ denote the capital and land holdings of agent $i$ and $P_t>0$ be the land price at date $t$. The budget constraint is
\begin{equation}
    c_{it}+k_{i,t+1}+P_tx_{it}=(A+1-\delta)k_{it}+(P_t+D)x_{i,t-1}. \label{eq:fof}
\end{equation}
It is convenient to define the beginning-of-period wealth by the right-hand side
\begin{equation}
    w_{it}\coloneqq (A+1-\delta)k_{it}+(P_t+D)x_{i,t-1}. \label{eq:wt}
\end{equation}
Because land is a safe asset, we may define the gross risk-free rate between time $t$ and $t+1$ by the return on land
\begin{equation}
    R_t\coloneqq \frac{P_{t+1}+D}{P_t}. \label{eq:riskfree}
\end{equation}

\paragraph{Equilibrium}
The economy starts at $t=0$ with some initial distribution of capital and land holdings $(k_{i0},x_{i,-1})_{i\in I}$. A \emph{rational expectations equilibrium} consists of a sequence of land prices $\set{P_t}_{t=0}^\infty$ and a stochastic process of individual choice variables $\set{(c_{it},k_{i,t+1},x_{it})_{i\in I}}_{t=0}^\infty$ such that
\begin{enumerate}
    \item each agent $i$ maximizes utility \eqref{eq:utility} subject to the budget constraint \eqref{eq:fof},
    \item the land market clears: $\int_I x_{it}\diff i=X$.
\end{enumerate}
Note that the sequence of land prices $\set{P_t}_{t=0}^\infty$ is deterministic because the economy features only idiosyncratic risk, which is diversified away at the aggregate level.

In any equilibrium, we may define aggregate consumption, capital stock, and wealth at date $t$ by $C_t\coloneqq \int_Ic_{it}\diff i$, $K_t\coloneqq \int_Ik_{it}\diff i$, and $W_t\coloneqq \int_I w_{it}\diff i$. We say that a rational expectations equilibrium associated with aggregate quantities $\set{(P_t,W_t,C_t,K_{t+1})}_{t=0}^\infty$ is a \emph{steady state} if the aggregate quantities are constant over time.

\subsection{Equilibrium analysis}
We characterize the equilibrium dynamics. Aggregating the individual budget constraint \eqref{eq:fof} and using the land market clearing condition $\int_I x_{it}\diff i=X$ yields the good market clearing condition
\begin{equation}
    C_t+K_{t+1}+P_tX=W_t\coloneqq (A+1-\delta)K_t+(P_t+D)X. \label{eq:Wt}
\end{equation}
Since the utility function is logarithmic, as is well known, the optimal consumption rule is $c_t=(1-\beta)w_t$. Aggregating across agents, we obtain
\begin{equation}
    C_t=(1-\beta)W_t. \label{eq:Ct}
\end{equation}
Combining \eqref{eq:Wt} and \eqref{eq:Ct}, we obtain aggregate savings
\begin{equation}
    K_{t+1}+P_tX=\beta W_t. \label{eq:savings}
\end{equation}
Let $\phi_t\ge 0$ be the fraction of aggregate savings flowing to capital, so
\begin{equation}
    (K_{t+1},P_tX)=(\beta \phi_tW_t,\beta(1-\phi_t)W_t). \label{eq:phit}
\end{equation}
Because there is no borrowing/lending and the fraction of entrepreneurs is $\pi\in (0,1)$, we must have $\phi_t\in [0,\pi]$.

Combining \eqref{eq:Wt}, \eqref{eq:Ct}, and \eqref{eq:phit}, we obtain the aggregate wealth dynamics
\begin{equation}
    (1-\beta+\beta\phi_t)W_t=\beta (A+1-\delta)\phi_{t-1} W_{t-1}+DX. \label{eq:W_dynamics}
\end{equation}
Using \eqref{eq:phit} to eliminate $W_t$ from \eqref{eq:W_dynamics}, we obtain the land price dynamics
\begin{equation}
    P_t=\frac{\beta(A+1-\delta)(1-\phi_t)}{1-\beta+\beta\phi_t}\frac{\phi_{t-1}}{1-\phi_{t-1}}P_{t-1}+\frac{\beta(1-\phi_t)}{1-\beta+\beta\phi_t}D. \label{eq:P_dynamics}
\end{equation}
Note that \eqref{eq:P_dynamics} also determines the dynamics of the price-rent ratio defined by $P_t/D$. Using \eqref{eq:riskfree} and \eqref{eq:P_dynamics}, the gross risk-free rate is
\begin{equation}
    R_t=\frac{\beta(A+1-\delta)(1-\phi_{t+1})}{1-\beta+\beta\phi_{t+1}}\frac{\phi_t}{1-\phi_t}+\frac{1}{1-\beta+\beta\phi_{t+1}}\frac{D}{P_t}. \label{eq:Rt}
\end{equation}
Because entrepreneurs can choose between capital and land investment, an arbitrage argument implies
\begin{equation}
    \phi_t\begin{cases*}
        =0 & if $R_t>A+1-\delta$,\\
        \in [0,\pi] & if $R_t=A+1-\delta$,\\
        =\pi & if $R_t<A+1-\delta$.
    \end{cases*}\label{eq:arbitrage}
\end{equation}

Interestingly, our model features both static and dynamic multiplier effects (positive feedback loops) between asset prices and the aggregate economy. To see why, by \eqref{eq:Wt} an increase in land price $P_t$ raises current wealth $W_t$, which increases the land price further through increasing aggregate savings \eqref{eq:phit}. This is the static multiplier effect. But an increase in the current wealth $W_t$ leads to high investment $K_{t+1}$ through \eqref{eq:phit}, which also increases the next period's wealth $W_{t+1}$ through \eqref{eq:Wt}. This is the dynamic multiplier effect. The presence of a positive feedback loop suggests that the dynamics of the model could qualitatively change depending on the level of productivity $A$. The following proposition provides a complete characterization of steady states.

\begin{prop}[Steady state]\label{prop:steady}
Define the thresholds for productivity
\begin{equation}
    \ubar{A}\coloneqq \frac{1-\beta}{\beta}+\delta<\frac{1-\beta}{\beta\pi}+\delta \eqqcolon \bar{A}. \label{eq:Abar}
\end{equation}
Then the following statements are true.
\begin{enumerate}
    \item\label{item:steady_L} If $A\le \ubar{A}$, the unique steady state gross risk-free rate is $R\coloneqq 1/\beta$ and land price equals the fundamental value $P=\frac{D}{R-1}$. If $A<\ubar{A}$, there is no capital investment: $\phi=0$. If $A=\ubar{A}$, $\phi\in [0,\pi]$ is indeterminate.
    \item\label{item:steady_M} If $\ubar{A}<A<\bar{A}$, the unique steady state gross risk-free rate is
    \begin{equation}
        R=\frac{1-\beta\pi(A+1-\delta)}{\beta(1-\pi)}, \label{eq:Rf}
    \end{equation}
    land price equals the fundamental value $P=\frac{D}{R-1}$, and there is full capital investment: $\phi=\pi$.
    \item\label{item:steady_H} If $A\ge \bar{A}$, there exist no steady states.
\end{enumerate}
\end{prop}

The case \ref{item:steady_L} with $A\le \ubar{A}$ is less interesting because only the real estate sector is active. In what follows, we focus on the case $A>\ubar{A}$. The following lemma shows that the capital-intensive sector is always active.

\begin{lem}\label{lem:invest}
If $A>\ubar{A}$, then $\phi_t>0$ for all $t$ in all equilibria.
\end{lem}

We next establish the existence of equilibrium and provide a characterization.

\begin{prop}\label{prop:exist}
If $A>\ubar{A}$, for any initial aggregate capital $K_0\ge 0$, there exists an equilibrium with eventual full capital investment: $\phi_t=\pi$ for all sufficiently large $t$. Consequently, for large enough $t$ the land price dynamics is given by
\begin{equation}
    P_t=\frac{\beta\pi(A+1-\delta)}{1-\beta+\beta\pi}P_{t-1}+\frac{\beta(1-\pi)}{1-\beta+\beta\pi}D. \label{eq:P_dynamics2}
\end{equation}
Furthermore, the following statements are true.
\begin{enumerate}
    \item\label{item:AL} If $A<\bar{A}$, then $R_t$ converges to $R_f=R$ in \eqref{eq:Rf} and $P_t$ converges to
    \begin{equation}
        P\coloneqq \frac{\beta(1-\pi)}{1-\beta-\beta\pi(A-\delta)}D=\frac{D}{R_f-1}. \label{eq:Pf}
    \end{equation}
    \item\label{item:AH} If $A\ge \bar{A}$, then $R_t$ converges to
    \begin{equation}
        R_b\coloneqq \frac{\beta\pi(A+1-\delta)}{1-\beta+\beta\pi}\ge 1 \label{eq:Rb}
    \end{equation}
    and $P_t$ diverges to $\infty$.
\end{enumerate}
\end{prop}

\subsection{Phase transition from balanced to unbalanced growth} 
In the bare-bones model, productivity $A$ can be interpreted as a parameter governing the degree of technological innovation in the manufacturing sector.\footnote{An increase in $A$ can be interpreted in two ways, \ie, product innovation and process innovation. If we interpret an increase in $A$ as an improvement in the quality of existing capital, it is product innovation. Another is to consider a situation in which $z$ captures the degree of efficiency in creating capital, and a unit of investment produces $z>1$ units of capital. This innovation can be interpreted as process innovation, in which case $z$ is mathematically included in $A$. In either cases, $A$ can be seen as the degree of technological innovation.} Propositions \ref{prop:steady} and \ref{prop:exist} show that our model can put the structural transformation in the economy in a historical context. When productivity is low ($A<\ubar{A}$), only the land sector is active and the economy is characterized by \emph{balanced growth}, where both the land price and rent grow at the same rate (they are constant). This economy can be interpreted as a Malthusian economy where land plays a major role for production. Once productivity exceeds the first threshold $\ubar{A}$, the first phase transition occurs and the manufacturing sector arises, which could be interpreted as Industrial Revolution or the birth of a modern economy where capital and knowledge play important roles for production. As long as productivity is not too high ($A<\bar{A}$), the land and manufacturing sectors coexist and the economy is still characterized by balanced growth. However, when productivity increases further and exceeds $\bar{A}$, the second phase transition occurs. The manufacturing sector starts growing at a faster rate than the land sector and the economy is characterized by \emph{unbalanced growth}. Once the economy enters this stage, land prices, pulled up by the growing manufacturing sector, rise at a faster rate than rents.

The following proposition shows that the asset pricing implications in worlds of balanced and unbalanced growth are markedly different.

\begin{prop}[Asset pricing]\label{prop:bubble_char}
Consider the equilibria in Propositions \ref{prop:steady} and \ref{prop:exist} with land price $P_t$ and fundamental value $V_t$. Then the following statements are true.
\begin{enumerate}
    \item\label{item:A<} If $A<\bar{A}$, there is no bubble: $P_t=V_t$ for all $t$ and the price-rent ratio $P_t/D$ converges.
    \item\label{item:A=} If $A=\bar{A}$, there is no bubble: $P_t=V_t$ for all $t$ and the price-rent ratio $P_t/D$ linearly diverges to $\infty$.
    \item\label{item:A>} If $A>\bar{A}$, there is a bubble: $P_t>V_t$ for all $t$ and the price-rent ratio $P_t/D$ exponentially diverges to $\infty$.
\end{enumerate}
\end{prop}

Note that the bubble existence condition $A>\bar{A}$ in Proposition \ref{prop:bubble_char}\ref{item:A>} is equivalent to the bubble necessity condition \eqref{eq:necessity}. To see why, the bubbleless gross risk-free rate $R$ is given by \eqref{eq:Rf}, the dividend growth rate is $G_d=1$, and using \eqref{eq:P_dynamics2}, the economic growth rate is given by
\begin{equation*}
    G\coloneqq \max\set{1,\frac{\beta\pi(A+1-\delta)}{1-\beta+\beta\pi}}.
\end{equation*}
It is straightforward to verify that the bubble necessity condition $R<G_d<G$ is equivalent to $A>\bar{A}$.

As a numerical example, we set $\pi=0.1$, $\beta=0.95$, $\delta=0.08$, and $D=1$. Figure \ref{fig:dynamics} shows the price dynamics \eqref{eq:P_dynamics}, one case for $A<\bar{A}$ (left) and another for $A>\bar{A}$ (right). It is clear that the strength of the static and dynamic multiplier effect characterized by the slope
\begin{equation*}
    \rho\coloneqq \frac{\beta\pi(A+1-\delta)}{1-\beta+\beta\pi}
\end{equation*}
determines the qualitative behavior of the model. The price sequence $\set{P_t}$ converges or diverges according as $\rho<1$ or $\rho\ge 1$.

\begin{figure}[!htb]
    \centering
    \includegraphics[width=0.48\linewidth]{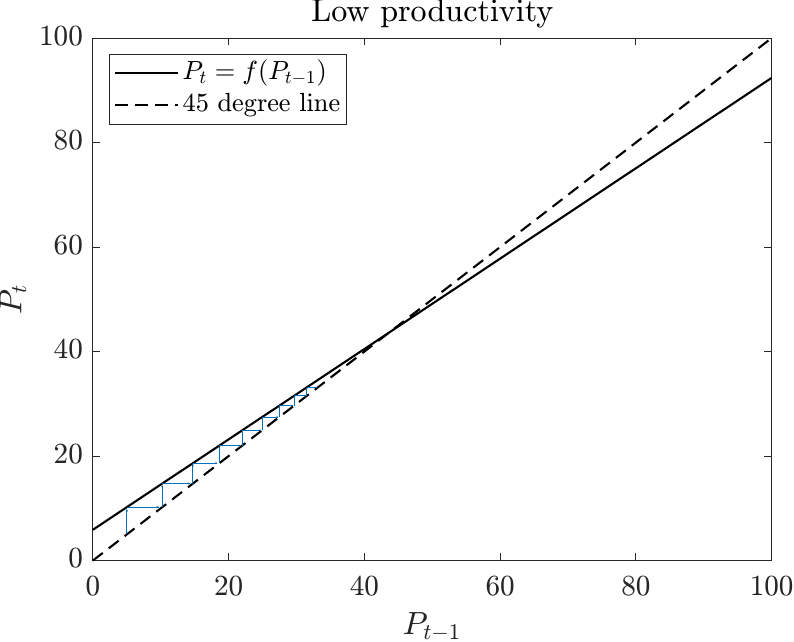}
    \includegraphics[width=0.48\linewidth]{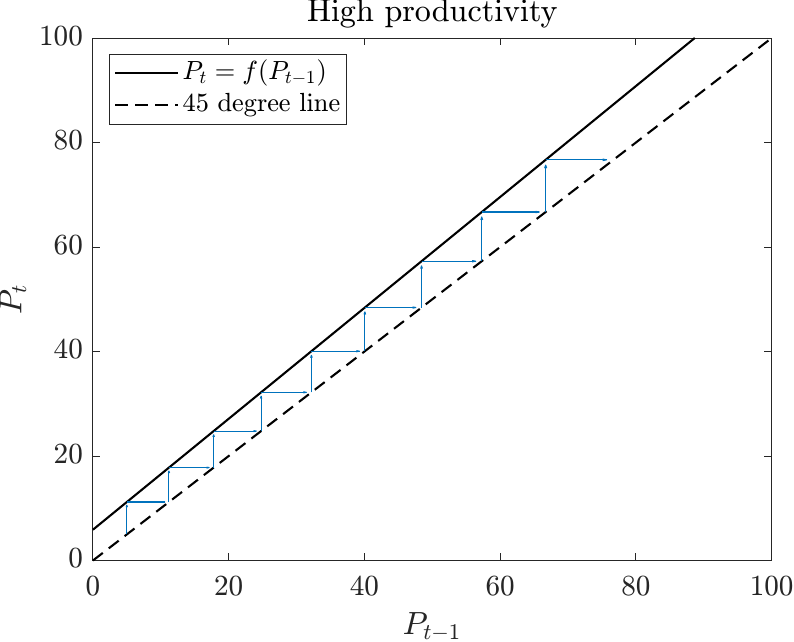}
    \caption{Land price dynamics.}
    \label{fig:dynamics}
    \caption*{\footnotesize Note: The figures show the price dynamics \eqref{eq:P_dynamics}. The parameter values are $\pi=0.1$, $\beta=0.95$, $\delta=0.08$, $D=1$, and $P_0=5$. productivity are either low ($A=0.4$, left panel) or high ($A=0.7$, right panel).}
\end{figure}

Figure \ref{fig:regime} shows the dependence of the long-run interest rate (the fundamental equilibrium interest rate $R_f=1/\beta$ if $A\le \ubar{A}$, $R_f$ in \eqref{eq:Rf} if $A\in (\ubar{A},\bar{A})$, or the bubbly equilibrium interest rate $R_b$ in \eqref{eq:Rb}) on productivity and the equilibrium land price regimes. As we increase productivity in the fundamental regime, the positive feedback loop between the land price and the aggregate economy gets stronger. Then the land price goes up and hence the interest rate goes down. However, because land pays constant rents, for the land price to be finite, the interest rate cannot fall below 1. When productivity exceeds the critical value that makes the interest rate equal to 1, a phase transition from the fundamental regime to the bubbly regime occurs, implying the emergence of land price bubbles under low interest rates. The fundamental regime is characterized by balanced growth, while the bubbly regime is characterized by unbalanced growth.

\begin{figure}[!htb]
    \centering
    \includegraphics[width=0.7\linewidth]{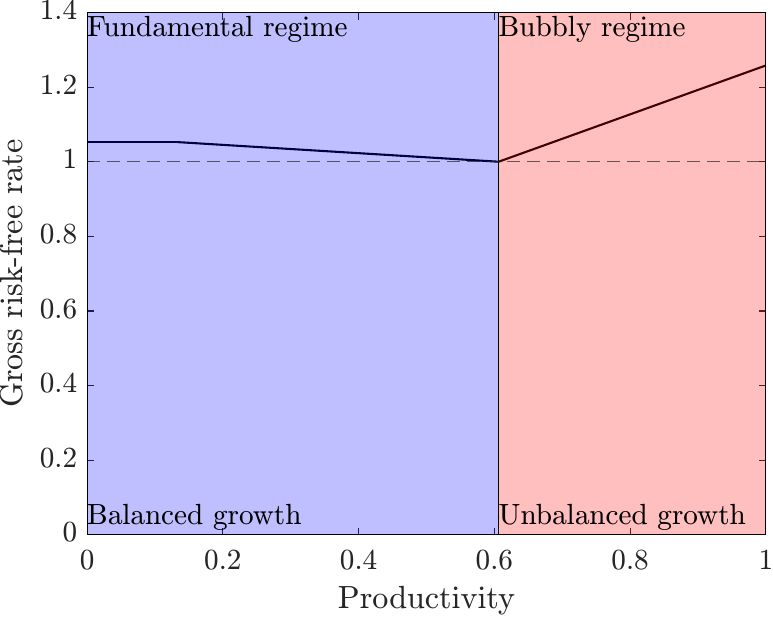}
    \caption{Phase transition of growth and equilibrium land price regimes.}
    \label{fig:regime}
    \caption*{\footnotesize Note: The figure shows the long-run interest rate ($R_f=1/\beta$ if $A\le \ubar{A}$, $R_f$ in \eqref{eq:Rf} if $A\in (\ubar{A},\bar{A})$, or $R_b$ in \eqref{eq:Rb}) against productivity $A$.}
\end{figure}

So far, we have placed land price bubbles in the context of major historical trends involving technological progress, but our analysis also has implications for rises and falls in land prices over a short period of time. If the fundamentals of the model suddenly change in an unexpected way for the agents, a bubble could emerge or collapse in a short period of time. To illustrate this point, Figure \ref{fig:boombust} shows the land price dynamics \eqref{eq:P_dynamics} when productivity unexpectedly changes. More specifically, the economy is initially in the steady state corresponding to $A=0.4$. At $t=1$, productivity unexpectedly increases to either 0.5 (fundamental regime, left) or 0.7 (bubbly regime, right), and at least when this shock happens, agents believe that it will stay so forever. However, at $t=11$, productivity unexpectedly reverts to $A=0.4$ and stays so forever.

\begin{figure}[!htb]
    \centering
    \includegraphics[width=0.48\linewidth]{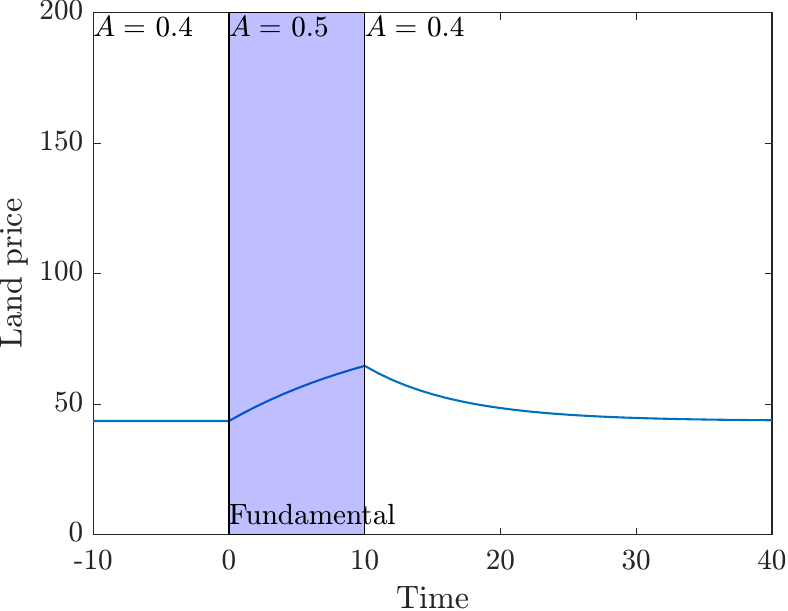}
    \includegraphics[width=0.48\linewidth]{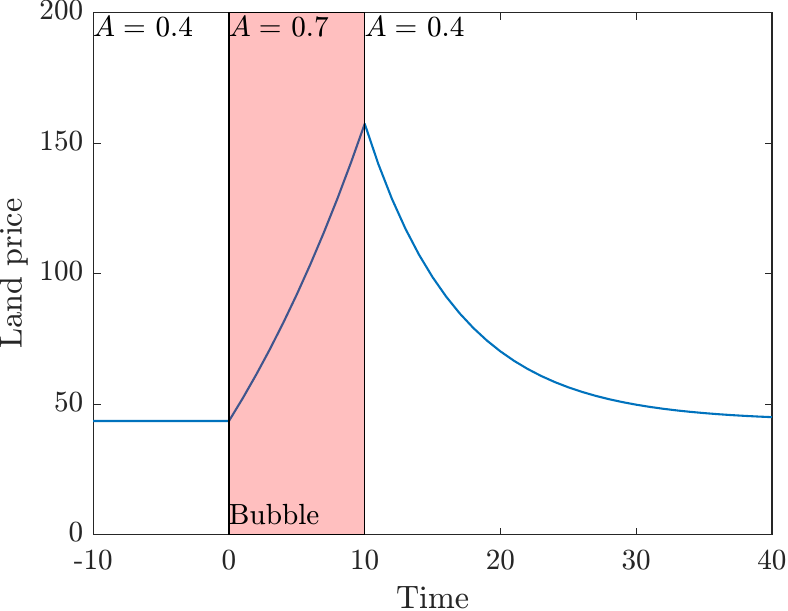}
    \caption{Boom and bust with varying productivity.}
    \label{fig:boombust}
    \caption*{\footnotesize Note: The figure shows the land price dynamics \eqref{eq:P_dynamics} when productivity $A$ unexpectedly changes.}
\end{figure}

In both cases, the land price exhibits a boom-bust cycle driven by the positive feedback loop between land price and the macroeconomy. However, the qualitative behavior is different depending on whether the boom is a bubble or not. In a bubble (right), land price grows exponentially (the price path is convex), whereas in a fundamental boom (left), land price converges to the new steady state (the price path is concave). Because in our model the rent is constant, the price-rent ratio shows the same pattern. This implies that the land price-rent ratio can be used as an indicator of land price bubbles, making it possible to connect our analysis with the bubble detection literature discussed in \S\ref{subsec:topic_criticism}.\footnote{It is possible to extend this bare-bones model to include aggregate risk, \ie, stochastic shocks to aggregate productivity. For instance, \citet{HiranoTodaUnbalanced} establish the Land Overvaluation Theorem in an OLG economy with aggregate risk. They show that land prices continue to fluctuate, always containing a bubble, with the bubble size expanding and contracting recurrently, which appears to be the emergence and collapse of land price bubbles.} 

\subsection{Technological progress and unbalanced growth} \label{sub:tech progress}

In the bare-bones model, we assumed that the productivity growth rate of capital is zero, that is, $A$ is constant. This is a typical assumption in the growth literature to ensure the existence of a steady state or a balanced growth path. Indeed, the well-known Uzawa Balanced Growth Theorem states that capital-augmenting technological progress is inconsistent with balanced growth \citep{Uzawa1961,Schlicht2006,JonesScrimgeour2008}. While standard macroeconomic models using the neoclassical production function are constructed to be consistent with the Uzawa Theorem, the conditions under which balanced growth occurs are known to be fragile.\footnote{\label{footnote:knife edge}Indeed, \citet*[p.~1306]{GrossmanHelpmanOberfieldSampson2017} note ``As with any model that generates balanced growth, knife-edge restrictions
are required to maintain the balance''. An exception to the growth rate restrictions is the Cobb-Douglas production function but it is obviously a knife-edge case with elasticity of substitution being exactly one.} In contrast to the Uzawa Theorem, there is broad empirical consensus that there is also capital-augmenting technical change \citep*{GrossmanHelpmanOberfieldSampson2017,GreenwoodHercowitskrusell1997,CaseyHorii2023}.

Given this observation, we now extend the bare-bones model so that productivity of both capital and land could be time-varying due to technological progress by replacing $A,D$ with $A_t,D_t$. The characterization of the global equilibrium dynamics remains the same, and \eqref{eq:P_dynamics} becomes
\begin{equation*}
    P_t=\frac{\beta\pi(A_t+1-\delta)}{1-\beta+\beta\pi}P_{t-1}+\frac{\beta(1-\pi)}{1-\beta+\beta\pi}D_t.
\end{equation*}
Dividing both sides by $D_t>0$ and letting $G_t\coloneqq D_t/D_{t-1}$ be the rent growth rate (land productivity growth), we obtain the global dynamics of the price-rent ratio
\begin{equation}
    \frac{P_t}{D_t}=\frac{\beta\pi(A_t+1-\delta)}{(1-\beta+\beta\pi)G_t}\frac{P_{t-1}}{D_{t-1}}+\frac{\beta(1-\pi)}{1-\beta+\beta\pi}. \label{eq:PD}
\end{equation}
By \eqref{eq:PD} and Lemma \ref{lem:bubble}, clearly the price-rent ratio will exponentially diverge to infinity and there is a bubble if
\begin{equation}
    \liminf_{t\to\infty} \frac{\beta\pi(A_t+1-\delta)}{(1-\beta+\beta\pi)G_t}>1. \label{eq:bubble_cond}
\end{equation}
For concreteness, suppose $G_t=G$ is constant. Then \eqref{eq:bubble_cond} is equivalent to
\begin{equation}
    \liminf_{t\to\infty} A_t>\frac{1-\beta}{\beta\pi}+G-1+\delta\eqqcolon \bar{A}(G), \label{eq:bubble_condA}
\end{equation}
which generalizes the condition $A>\bar{A}$ in Proposition \ref{prop:bubble_char}. In reality, the productivity of capital, whatever its growth rate, grows due to technological innovations. If capital productivity is expected to eventually exceed the threshold $\bar{A}(G)$, the economy exhibits unbalanced growth and is always in the transition process of the global dynamics associated with land bubbles.\footnote{In the case of unbalanced growth, it is more natural to think of the economy as being in a transition process of global dynamics all along. Indeed, the literature on structural transformation considers so, because as $t\to\infty$, the service sector will almost become dominant and other sectors such as agriculture and manufacturing will asymptotically vanish. In this sense, transitional dynamics would be more relevant, rather than the limiting situation.}

There are two important implications to be drawn from our analysis. First, our analysis provides new insights to the methodology of macroeconomic theory. Economists have long been trained and got accustomed to constructing models so that a steady state emerges. Perhaps because of that, the fact that we need knife-edge restrictions to obtain a steady state may have been overlooked. What our analysis shows is that even with a slightest bit of capital-augmenting technological progress, \ie, if we remove knife-edge restrictions, a steady-state equilibrium will no longer exist, generating unbalanced growth and asset price bubbles. 

Second, in sharp contrast with the common view that the land price should reflect its fundamental value along the long-run trend, the current model as well as those of \citet{HiranoTodaUnbalanced, HiranoTodaHousingbubble} show that asset price bubbles occur in the context of historical trends involving shifts in industrial structure, including the transition from the Malthusian economy to the modern economy via Industrial Revolution. This implies that the development of capitalist economies and asset price bubbles are tightly connected. Instead of pure bubbles, considering bubbles attached to real assets allows us to get to the essence of asset price bubbles from a historical perspective.

\section{Concluding remarks and open issues}

The seminal paper by \citet{Samuelson1958} has produced a vast literature on rational bubbles. This literature has made a significant progress with the recent development on financial factors and asset bubbles since the 2008 financial crisis. While the existing pure bubble models are very useful for understanding money or monetary economies, there have been significant limitations in using those models for analyzing bubbles attached to real assets such as stocks, housing, and land. As we have explained in \S\ref{subsec:topic_criticism}, in addition to the criticism that pure bubble models are unrealistic for describing bubbles attached to real assets, there are several other criticisms, which make applications and further development of the literature challenging.

To the best of our knowledge, the example in \citet[\S7]{Wilson1981} is the first one to show bubbles attached to dividend-paying assets as well as the nonexistence of fundamental equilibria. Although \citet{Wilson1981} provided this surprising result, he only gave an example in a fairly limited setting. Hence, it is not obvious to what extent there is generality and how relevant the result is, nor is it obvious what insightful implications can be drawn from a macroeconomic perspective. 

A series of our working papers \citep*{HiranoJinnaiTodaLeverage,HiranoTodaNecessity,HiranoTodaUnbalanced, HiranoTodaHousingbubble} have generalized \citet{Wilson1981}'s insight and uncovered the essence of asset price bubbles from a macroeconomic perspective. They have derived three new insights. First, they propose a conceptually new idea of the necessity of asset price bubbles and establish the Bubble Necessity Theorem within workhorse macroeconomic models (see \S\ref{subsec:found_necessary} for the concept of the necessity of asset bubbles). Second, they show the tight link between unbalanced growth and asset price bubbles, and find that asset pricing implications under balanced growth and unbalanced growth are markedly different. This implies that the essence of asset price bubbles attached to dividend-paying assets is nonstationarity. Third, they show that asset price bubbles occur within larger historical trends involving shifts in industrial structure that cause unbalanced growth. 

These findings have two important implications. First, in benchmark modern macro-finance models, it seems that there is a presupposition that asset prices \emph{must} reflect fundamentals. Indeed, it is fair to say that those models are constructed so that asset prices equal the fundamental values. Contrary to this view, we find that asset prices \emph{cannot} equal fundamental values. In other words, there are benchmark cases in which the notion that asset prices should basically reflect fundamentals is wrong. Second, bubbles attached to dividend-paying assets cannot arise in stationary models. To understand the essence of asset price bubbles, we need to depart from stationary models with a steady state to nonstationary models without it. Since economists have long been trained and accustomed to studying stationary models, if a steady state equilibrium does not exist, they tend to think that the model is broken or is a failure, and they tend to make assumptions or change the model setting so that a steady state is produced. We need to change this conventional way of thinking. The obsession with the steady-state model may have led to the common belief that bubbles do not occur in standard macroeconomic models or there is a fundamental difficulty in generating asset bubbles in dividend-paying assets.

Finally, we would like to discuss future directions. The theory of rational asset price bubbles attached to real assets has the potential to fundamentally change the conventional thinking about asset bubbles. Moreover, it remains largely unexplored. Here, we would like to raise some directions as examples. First, understanding why stock, land, and housing price bubbles occur in the context of history will be of great significance not only in terms of economic theory but also in terms of economic history. Tracing the origins of asset price bubbles suggests that they are closely related to the origin and development of capitalist economies. This direction could create a new research field connecting historical research and asset price bubbles. 

Second, as we mentioned in \S\ref{sub:tech progress}, by imposing knife-edge restrictions, standard macroeconomic models are constructed so that balanced growth, or a steady state, arises. If we extend those models to allow unbalanced growth dynamics by removing these restrictions, we expect that asset pricing implications will be markedly different, allowing for macro-finance analysis with a focus on asset price bubbles.

Third, asset price bubbles are also directly related to policy research. As the bare-bones model in \S\ref{sec:bare-bones} illustrates, the land price-rent ratio tells us important information about whether land prices contain a bubble ore not. Connecting this result to the bubble detecting literature in econometrics would lead to the identification of bubbles and the construction of early warning indicators of bubbles from both econometric and macroeconomic perspectives. Also, in the bare-bones model, we abstract from many realistic elements such as wage rigidity, price stickiness, defaults, the zero lower bound constraint, and further heterogeneity among economic agents. Hence, the model can be fleshed out by introducing those realistic elements and extended in various ways toward policy and quantitatively oriented analyses such as heterogeneous-agent New Keynesian models. 

Thus, bubble research advanced along the lines of \citet{Wilson1981} constitutes an untapped field with enormous potential for applications, just as \citet{Samuelson1958}'s paper has produced a vast literature on monetary economics. We hope that our review article will lead to fruitful outcomes.

\appendix

\section{Proof of \S\ref{sec:bare-bones} results}

\begin{proof}[Proof of Proposition \ref{prop:steady}]
Let $R_k\coloneqq A+1-\delta$ be the gross return on capital. By \eqref{eq:P_dynamics}, $(P,\phi)$ is part of a steady state if and only if
\begin{equation}
    P=\frac{\beta(1-\phi)}{1-\beta+\beta\phi(1-R_k)}D>0 \label{eq:P_steady}
\end{equation}
and \eqref{eq:arbitrage} holds. Using \eqref{eq:P_steady}, the gross risk-free rate \eqref{eq:Rt} becomes
\begin{equation}
    R=\frac{1-\phi\beta R_k}{\beta(1-\phi)}. \label{eq:R_steady}
\end{equation}
Using \eqref{eq:R_steady}, we obtain
\begin{equation}
    R\le R_k\iff R_k\ge 1/\beta\iff A\ge \frac{1-\beta}{\beta}+\delta=\ubar{A}. \label{eq:R_equiv}
\end{equation}
If $A<\ubar{A}$, then $R_k<R$ by \eqref{eq:R_equiv} and hence $\phi=0$. Then $R=1/\beta$ by \eqref{eq:R_steady}. If $A=\ubar{A}$, then $R=R_k=1/\beta$ by \eqref{eq:R_equiv}. In either case, we have $R=1/\beta$ and $P=\frac{\beta D}{1-\beta}=\frac{D}{R-1}$ by \eqref{eq:P_steady}, so the steady state interest rate and land price are unique and statement \ref{item:steady_L} holds.

In what follows, assume $A>\ubar{A}$. Then \eqref{eq:R_equiv} implies $R_k>R$ and hence $\phi=\pi$ by \eqref{eq:arbitrage}. By \eqref{eq:P_steady}, a (necessarily unique) steady state exists if and only if
\begin{equation*}
    1-\beta+\beta\pi(1-R_k)>0\iff A<\frac{1-\beta}{\beta\pi}+\delta=\bar{A}.
\end{equation*}
Setting $\phi=\pi$ in \eqref{eq:R_steady}, we obtain \eqref{eq:Rf}. Furthermore, \eqref{eq:P_steady} implies $P=\frac{D}{R-1}$. Therefore both statements \ref{item:steady_M} and \ref{item:steady_H} hold.
\end{proof}

\begin{proof}[Proof of Lemma \ref{lem:invest}]
Suppose to the contrary that $\phi_t=0$ for some $t$. Then \eqref{eq:P_dynamics} implies
\begin{equation*}
    P_t\ge \frac{\beta}{1-\beta}D\iff \frac{D}{P_t}\le \frac{1-\beta}{\beta}.
\end{equation*}
Therefore \eqref{eq:Rt} implies
\begin{equation*}
    R_t=\frac{1}{1-\beta+\beta\phi_{t+1}}\frac{D}{P_t}\le \frac{1}{1-\beta}\frac{1-\beta}{\beta}=\frac{1}{\beta}<A+1-\delta
\end{equation*}
because $A>\ubar{A}$. Then \eqref{eq:arbitrage} implies $\phi_t=\pi$, which is a contradiction.
\end{proof}

To prove Proposition \ref{prop:exist}, we establish a series of lemmas.

\begin{lem}\label{lem:W0}
If initial aggregate wealth satisfies
\begin{equation}
    W_0\ge \bar{W}_0\coloneqq \frac{\beta(1-\pi)DX}{(1-\beta)(A+1-\delta)}, \label{eq:W0bar}
\end{equation}
then $\phi_t=\pi$ for all $t$ is an equilibrium path.
\end{lem}

\begin{proof}
Set $\phi_t=\pi$ for all $t$. Using \eqref{eq:phit}, the gross risk-free rate \eqref{eq:Rt} becomes
\begin{equation*}
    R_t=\frac{\beta\pi(A+1-\delta)}{1-\beta+\beta\pi}+\frac{\beta(1-\pi)}{1-\beta+\beta\pi}\frac{DX}{W_t}\le A+1-\delta,
\end{equation*}
where the last inequality follows from \eqref{eq:arbitrage}. Solving the inequality, we obtain
\begin{equation*}
    W_t\ge \frac{\beta(1-\pi)DX}{(1-\beta)(A+1-\delta)}=\bar{W}_0.
\end{equation*}
To complete the proof, it remains to show that if $W_0\ge \bar{W}_0$, then $W_t\ge \bar{W}_0$ for all $t$ if $\set{W_t}$ satisfies \eqref{eq:W_dynamics} with $\phi_t=\pi$ for all $t$. Write \eqref{eq:W_dynamics} as
\begin{equation*}
    W_t=f(W_{t-1})\coloneqq \frac{\beta\pi(A+1-\delta)}{1-\beta+\beta\pi}W_{t-1}+\frac{DX}{1-\beta+\beta\pi}.
\end{equation*}
Since $f$ is increasing, it suffices to show $f(\bar{W}_0)\ge \bar{W}_0$. Using $A>\ubar{A}$ and hence $\beta(A+1-\delta)>1$, it follows from \eqref{eq:W0bar} that
\begin{align*}
     f(\bar{W}_0)-\bar{W}_0&\ge \frac{DX}{1-\beta+\beta\pi}\left(\frac{\beta^2\pi(1-\pi)}{1-\beta}+1\right)-\frac{\beta^2(1-\pi)DX}{1-\beta}\\
     &=\frac{DX}{1-\beta+\beta\pi}(1-\beta^2(1-\pi)^2)>0. \qedhere
\end{align*}
\end{proof}

\begin{lem}\label{lem:global}
In the equilibrium in Lemma \ref{lem:W0}, the price dynamics \eqref{eq:P_dynamics2} holds. Furthermore, the statements \ref{item:AL} and \ref{item:AH} in Proposition \ref{prop:exist} hold.
\end{lem}

\begin{proof}
Consider an equilibrium described in Lemma \ref{lem:W0}. Setting $\phi_t=\pi$ in \eqref{eq:P_dynamics}, we obtain the price dynamics \eqref{eq:P_dynamics2}, which we write as $P_t=\rho P_{t-1}+\gamma$ for
\begin{equation}
    (\rho,\gamma)=\left(\frac{\beta\pi(A+1-\delta)}{1-\beta+\beta\pi},\frac{\beta(1-\pi) D}{1-\pi+\beta\pi}\right). \label{eq:alphagamma}
\end{equation}

\ref{item:AL} If $A<\bar{A}$, then $\rho\in (0,1)$. Then $\set{P_t}$ satisfying \eqref{eq:P_dynamics2} converges to $P\coloneqq \gamma/(1-\rho)$, which is \eqref{eq:Pf}. Then $R_t$ converges to
\begin{equation*}
    R=\frac{P_f+D}{P_f}=\frac{1-\beta\pi(A+1-\delta)}{\beta(1-\pi)},
\end{equation*}
which is \eqref{eq:Rf}.

\ref{item:AH} If $A\ge \bar{A}$, then $\rho\ge 1$. Then $\set{P_t}$ satisfying \eqref{eq:P_dynamics2} diverges to $\infty$. Using \eqref{eq:Rt}, $R_t$ converges to $R_b$ in \eqref{eq:Rb}.
\end{proof}

\begin{proof}[Proof of Proposition \ref{prop:exist}]
By shifting time if necessary, we seek an equilibrium in which $\phi_t<\pi$ for all $t<0$ and $\phi_t=\pi$ for all $t\ge 0$. Given $W_0\ge \bar{W}_0$, for $t\ge 0$, we may construct an equilibrium path using Lemmas \ref{lem:W0} and \ref{lem:global} and the conclusion holds. Therefore to complete the proof, it suffices to extend $\set{W_t}$ for $t<0$ consistent with the initial aggregate capital. To construct an equilibrium path including $t<0$, fix $W_0\ge \bar{W}_0$. We only provide a sketch because the complete proof is tedious.

We first construct $\set{W_t}$. By Lemma \ref{lem:invest}, we have $\phi_t\in (0,\pi)$ for $t<0$, so \eqref{eq:arbitrage} implies $R_t=R_k\coloneqq A+1-\delta$. Since the risk-free rate equals the return on capital, it follows from \eqref{eq:Ct} that $W_{t+1}=\beta R_k W_t$ for $t<0$. Hence given $W_0$, we may extend $\set{W_t}$ for $t=-j<0$, namely $W_{-j}=(\beta R_k)^{-j}W_0$.

We next construct $\set{\phi_t}$. Noting that $R_t=R_k$ for $t<0$, we obtain
\begin{equation*}
    R_k=R_t=\frac{P_{t+1}+D}{P_t}\iff P_tX=\frac{1}{R_k}(P_{t+1}X+DX).
\end{equation*}
Dividing both sides by $\beta W_t=W_{t+1}/R_k$ and using \eqref{eq:phit}, we obtain
\begin{equation*}
    1-\phi_t=\beta(1-\phi_{t+1})+\frac{DX}{W_{t+1}}.
\end{equation*}
Because $W_t$ for $t\le 0$ is already constructed and $\phi_0=\pi$, we may recursively determine $\phi_t$ for $t=-j<0$. Specifically, we obtain
\begin{equation}
    1-\phi_{-j}=\beta^j(1-\pi)+\beta^{j-1}\frac{DX}{W_0}(1+R_k+\dots+R_k^{j-1}). \label{eq:phi-j}
\end{equation}

Finally we derive the equilibrium condition. If the equilibrium dynamics starts at $t=-j$ with initial aggregate wealth $K\ge 0$ given, using \eqref{eq:Wt}, the initial wealth is
\begin{equation*}
    W_{-j}=R_kK+(P_{-j}+D)X=R_kK+\beta(1-\phi_{-j})W_{-j}+DX.
\end{equation*}
Using $W_{-j}=(\beta R_k)^{-j}W_0$ and \eqref{eq:phi-j}, we may simplify the equation as
\begin{equation*}
    (1-\beta^{j+1}(1-\pi))W_0=\beta^j(R_k^{j+1}K+DX(1+R_k+\dots+R_k^j)).
\end{equation*}
Imposing $W_0/(\beta R_k)=W_{-1}<\bar{W}_0\le W_0$ and solving for $j$, we obtain an equilibrium.
\end{proof}

\begin{proof}[Proof of Proposition \ref{prop:bubble_char}]
If $A<\bar{A}$, then $\rho<1$ in \eqref{eq:P_dynamics2} and $\set{P_t}$ converges. Because the dividend yield $D/P_t$ converges to a positive constant, we have $\sum_{t=1}^\infty D/P_t=\infty$, so there is no bubble by Lemma \ref{lem:bubble}.

If $A=\bar{A}$, then $\rho=1$ in \eqref{eq:P_dynamics2} and hence $P_t=P_0+\gamma t$. Since $P_t$ grows linearly, we have $\sum_{t=1}^\infty D/P_t=\infty$, so there is no bubble by Lemma \ref{lem:bubble}.

If $A>\bar{A}$, then $P_t$ asymptotically grows at rate $\rho>1$. Since $D/P_t\sim \rho^{-t}$, we have $\sum_{t=1}^\infty D/P_t<\infty$, so there is a bubble by Lemma \ref{lem:bubble}.
\end{proof}

\printbibliography

\section{Figures (not for publication)}

\begin{figure}[!htb]
\centering
\begin{subfigure}{0.48\linewidth}
    \includegraphics[width=\linewidth]{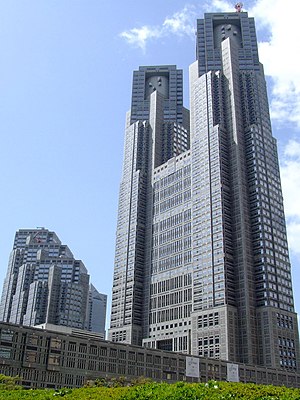}
    \caption{Metropolitan Government Building.}
\end{subfigure}
\begin{subfigure}{0.48\linewidth}
    \includegraphics[width=\linewidth]{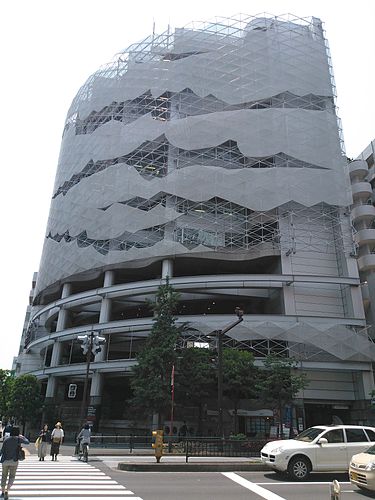}
    \caption{Joule A.}
\end{subfigure}
\caption{Towers of Bubble (Images from Wikipedia).}
\label{fig:tower}
\end{figure}

\end{document}